\newif\ifarxiv
\arxivtrue

\ifarxiv
  \documentclass[11pt,a4paper]{article}
\else
  \RequirePackage{fix-cm}
  \documentclass[smallextended,envcountsect,envcountsame]{svjour3}
\fi

\usepackage[T1]{fontenc}
\usepackage[utf8]{inputenc}
\usepackage{amsmath,amssymb,mathtools}
\usepackage{booktabs}
\usepackage{graphicx}
\usepackage[expansion=false]{microtype}
\usepackage[square,numbers,sort&compress]{natbib}

\ifarxiv
  \usepackage{amsthm}
  \usepackage[margin=1.05in]{geometry}
  
  \theoremstyle{plain}
  \newtheorem{theorem}{Theorem}[section]
  \newtheorem{proposition}[theorem]{Proposition}
  \newtheorem{lemma}[theorem]{Lemma}
  \newtheorem{corollary}[theorem]{Corollary}
  \theoremstyle{definition}
  \newtheorem{definition}[theorem]{Definition}
  \newtheorem{assumption}[theorem]{Assumption}
  \newtheorem{remark}[theorem]{Remark}
  
\else
  \smartqed
  \journalname{Finance and Stochastics}
  \providecommand{\proofname}{Proof}
  \renewenvironment{proof}[1][\proofname]{%
    \par\addvspace{\medskipamount}%
    \noindent\emph{#1.}\ \ignorespaces}%
   {\qed\par\addvspace{\medskipamount}}
  \spnewtheorem{assumption}[theorem]{Assumption}{\bfseries}{\rmfamily}
  \def\JELclassname{{\bfseries JEL Classification}\enspace}
  \def\JELclass#1{\par\addvspace\medskipamount{\rightskip=0pt plus1cm
  \def\and{\ifhmode\unskip\nobreak\fi\ $\cdot$
  }\noindent\JELclassname\ignorespaces#1\par}}
\fi

\usepackage[hidelinks]{hyperref}

\newcommand{\R}{\mathbb{R}}
\newcommand{\E}{\mathbb{E}}
\newcommand{\Prob}{\mathbb{P}}
\newcommand{\Qmeas}{\mathbb{Q}}

\newcommand{\Lop}{\mathcal{L}}
\newcommand{\vk}{\mathbf{v}_\kappa}

\newcommand{\papertitle}{Optimal Liquidation with Support and Resistance
Levels under Multi-Skew Brownian Motion}
\newcommand{\disclaimershort}{Any errors are the author's own.}
\newcommand{\disclaimerlong}{This paper reports the author's personal research
and is unrelated to his employment; the ideas presented here have no
connection with Mizuho Securities Co., Ltd. Any errors are the author's own.}

\ifarxiv
  \title{\bfseries \papertitle}
  \author{Jun Maeda\thanks{\disclaimershort}\\[4pt]
  \normalsize \texttt{jun.maeda@warwickgrad.net}}
  \date{September 2026}
\else
  \title{\papertitle\thanks{\disclaimerlong}}
  \titlerunning{Optimal liquidation with support and resistance levels}
  \author{Jun Maeda}
  \authorrunning{J. Maeda}
  \institute{J. Maeda \at
             Mizuho Securities Co., Ltd., Tokyo, Japan \\
             \email{j.maeda@mizuho-sc.com} \\
             Corresponding author; proofs to this address.}
  \date{Received: date / Accepted: date}
\fi

\begin{document}

\ifarxiv
  \pagestyle{myheadings}
  \markright{\small Optimal liquidation with support and resistance levels}
  \maketitle
  \thispagestyle{plain}
\fi

\begin{abstract}
We solve the perpetual liquidation problem for a geometric multi-skew
Brownian motion carrying local-time pushes upward at a support level and
downward at a resistance level, a model of technical analysis that is Markov
in the price alone. Three geometries arise, separated by a closed-form
criterion: the continuation band lies below resistance, straddles it, or pins
its upper edge there. Selling into resistance is a conclusion rather than an
assumption. Identification is asymmetric: exercise behaviour determines the
support parameters exactly, while the strength of resistance is recoverable
only on an interval with a closed-form endpoint.
\end{abstract}

\ifarxiv

\fi

\section{Introduction}
\label{sec:intro}

Technical analysis holds that certain price levels are special. A
\emph{support} level is one at which a falling asset tends to be bought, a
\emph{resistance} level one at which a rising asset tends to be sold. The
claim has a precise probabilistic content: the law of the price is not
invariant under translation of the state space, and the asymmetry is
concentrated on a set of Lebesgue measure zero. It is not only folklore:
\citet{Osler00} finds that the support and resistance levels published by
foreign-exchange firms do predict intraday reversals, and \citet{Osler03}
attributes the effect to the clustering of limit and stop-loss orders at
round numbers --- a mechanism that acts at a level rather than over an
interval.

Skew Brownian motion, introduced by \citet{ItoMcKean63} and characterised as
the unique strong solution of an equation with a local-time term by
\citet{HarrisonShepp81}, is the minimal modification of Brownian motion that
makes one level special; see \citet{Lejay06} for a survey of its
constructions. \citet{AlvarezSalminen17} solved the perpetual stopping problem
for it and found two phenomena that recur throughout the present paper: the
skew point always lies in the continuation region for a large class of
pay-offs, and the exercise region can be disconnected even for a linear
pay-off.

Continuous-time treatments of technical trading itself are fewer.
\citet{Lorig19} take the rules as given and compute the expected pay-off of
moving-average strategies under a general diffusion. \citet{DeAngelisPeskir16}
leave the price a geometric Brownian motion and place the level in the
objective rather than in the dynamics: the holder has an aspiration level and
stops so as to minimise the expected distance of the price from it, so that
support and resistance are targets to be predicted, not features of the
law. In both the level or the rule is an input. Here it is an output: the
levels are written into the dynamics, and the trading rule is what the
stopping problem returns.

The existing models of support and resistance carry a memory. In
\citet{JackaMaeda20} a flag records whether the asset is in a rising or a
falling regime, and switches when the price crosses thresholds on either side
of a reference level; \citet{Henderson26} use the same device, with a
partially reflecting barrier at the level whose coefficients depend on the
regime. A single level can then act as support in one regime and as
resistance in the other, and it changes role as the price moves through a band
around it. This role reversal --- broken support becomes resistance --- is
path-dependent by nature: whether the level supports or resists depends on
where the price has been.

This paper studies the complementary situation, in which support and
resistance are distinct levels with fixed roles, both active at once. We keep
the local-time mechanism and place a permanent upward push at a support level
$L$ and a permanent downward push at a resistance level $H>L$. The resulting
process is a geometric multi-skew Brownian motion in the sense of
\citet{Ramirez11}, Markov in the price alone, with one drift and one
volatility and no auxiliary state. We ask what this configuration implies for
the holder's liquidation decision, and what that decision reveals about the
two levels.

\subsection{Why the physical measure}
\label{sec:whyP}

A price process carrying local time in its returns is not a risk-neutral
object. \citet{Rossello12} showed that such models admit arbitrage, a point
the option-pricing literature largely set aside, and
\citet{Torricelli26} argue that the process used in that literature is not the
It\^o--McKean one. Proposition~\ref{prop:noEMM} records the elementary reason:
local time is singular with respect to Lebesgue measure and survives every
equivalent change of measure, so $e^{-rt}S_t$ cannot be made a local
martingale unless the skew vanishes.

We therefore make no valuation claim. The problem is the \emph{timing} problem
of an agent holding one unit of the asset, posed under the physical measure
with a subjective discount rate $r>0$ representing opportunity cost or
funding. No replication, measure change or derivative price appears below, and
the admissible strategies are single stopping times: the arbitrage of
\citet{Rossello12} requires continuous trading through a skew point, which a
stopping time cannot execute.

Working under $\mathbb P$ is not a precaution but a necessity, for two
independent reasons. The first is Proposition~\ref{prop:noEMM}: no equivalent
measure turning $e^{-rt}S_t$ into a local martingale exists once $\beta\ne0$.
The second holds even if one sets that aside. Under such a measure the drift
would be pinned at $r$, so $e^{-rt}S_t$ would be a martingale, and optional
sampling would give $\E\bigl[e^{-r\tau}S_\tau\bigr]=S_0$ for every stopping
time $\tau$: every liquidation rule would be optimal and the problem would
carry no information. A timing problem exists at all only because the drift is
free of $r$, which is to say only under the physical measure.

\subsection{Outline}

Section~\ref{sec:model} sets up the model. Section~\ref{sec:prelim} gives the
fundamental solutions, the crossing matrix and the smooth-fit lemma, together
with the one-zero property of the fundamental pair
(Lemma~\ref{lem:onezero}) on which the later comparison arguments rest.
Section~\ref{sec:fb} solves the free boundary and establishes the three
geometries, the exact irrelevance of $\beta_H$ in the first
(Theorem~\ref{thm:Airrelevant}), the closed-form criterion for the third
(Theorem~\ref{thm:corner}) and the continuity of the transition
(Proposition~\ref{prop:transition}). Section~\ref{sec:verify} proves
dominance and completes the verification. Section~\ref{sec:ident} contains the
identification results: resistance is identified only on an explicit interval
(Theorem~\ref{thm:ident}), support exactly and in closed form
(Theorem~\ref{thm:inversion}). Section~\ref{sec:samesign} treats two skews of the same
sign: the problem decouples into two bands with a three-component stopping
region (Proposition~\ref{prop:decouple}), the bands merge continuously at an
explicit critical separation (Proposition~\ref{prop:merger}), and dominance on
the merged band --- where the argument of Section~\ref{sec:verify} breaks
down --- is recovered by comparison with the single-skew problem
(Theorem~\ref{thm:domsame}), making the criterion sharp. Both skew parameters
are then identified. Section~\ref{sec:othersigns} disposes of the remaining
two sign pairs: with both skews negative there is no continuation region, and
with $\beta_L<0<\beta_H$ the geometry of Sections~\ref{sec:fb}--\ref{sec:verify}
recurs with the two levels exchanged, the corner becoming a stop-loss at $L$.
Section~\ref{sec:disc} concludes, summarising the results and discussing the
mechanism behind them and open
directions. Appendix~\ref{app:proofs} carries two long proofs and
Appendix~\ref{app:num} the numerical validation.

\section{The model}
\label{sec:model}

\subsection{Geometric multi-skew Brownian motion}

Let $(\Omega,\mathcal F,(\mathcal F_t),\Prob)$ carry a standard Brownian
motion $W$. Fix levels $0<L<H$, parameters $\mu\in\R$, $\sigma>0$, and skew
coefficients
\begin{equation}
\beta_L\in(0,1),\qquad \beta_H\in(-1,0).
\label{eq:betasigns}
\end{equation}
Sections~\ref{sec:samesign} and~\ref{sec:othersigns} relax
\eqref{eq:betasigns} and treat the other three sign pairs; everything before
them assumes it.
The price solves
\begin{equation}
dS_t=\mu S_t\,dt+\sigma S_t\,dW_t
 +\frac{\beta_L}{L}S_t\,d\ell^L_t(S)+\frac{\beta_H}{H}S_t\,d\ell^H_t(S),
\qquad S_0=x>0,
\label{eq:SDE}
\end{equation}
with $\ell^L,\ell^H$ the symmetric local times of $S$. Equivalently, with
$X=\log S$, $\nu=\mu-\tfrac12\sigma^2$, $\lambda=\log L$, $\eta=\log H$,
\begin{equation}
dX_t=\nu\,dt+\sigma\,dW_t+\beta_L\,d\ell^\lambda_t(X)
     +\beta_H\,d\ell^\eta_t(X),
\label{eq:SDElog}
\end{equation}
so $X$ is a multi-skew Brownian motion with drift and $S=e^X$.

The sign convention \eqref{eq:betasigns} is the whole of the economics:
$\beta_L>0$ pushes the price up at $L$, which is what a support level does,
and $\beta_H<0$ pushes it down at $H$. There are no regimes: a level's
character is fixed once and for all, and the process has no memory of where it
has been. Remark~\ref{rem:signload} shows the convention is load-bearing
rather than cosmetic.

\begin{proposition}[Well-posedness]
\label{prop:wellposed}
For $|\beta_L|,|\beta_H|\le1$ equation \eqref{eq:SDElog} admits a unique
strong solution, and $S=e^X$ is a regular strong Markov diffusion on
$(0,\infty)$.
\end{proposition}

\begin{proof}
\citet{HarrisonShepp81} give strong existence and uniqueness for a single
local-time term with $|\beta|\le1$ (and non-existence for $|\beta|>1$);
\citet{LeGall84} covers one-dimensional equations involving the local time of
the unknown process with drift, and \citet{Ramirez11} treats finitely many
skew points. Since $\lambda<\eta$ the two local-time terms are supported on
disjoint sets and the solution is obtained by concatenating single-skew
solutions across excursions between the levels. Regularity and the strong
Markov property are inherited from the scale and speed characterisation, the
scale function being continuous and strictly increasing with a kink at each
skew point.
\end{proof}

\begin{lemma}[Transmission conditions]
\label{lem:transmission}
Write
\[
\Lop u(x)=\tfrac12\sigma^2x^2u''(x)+\mu xu'(x),\qquad x\notin\{L,H\}.
\]
A continuous function $u$ that is $C^2$ on each of $(0,L]$, $[L,H]$ and
$[H,\infty)$ --- that is, $C^2$ off $\{L,H\}$ with $u'$ and $u''$
admitting finite one-sided limits at $L$ and $H$ --- lies in the domain of the
generator of \eqref{eq:SDE} if and only if
\begin{equation}
(1+\beta_L)u'(L^+)=(1-\beta_L)u'(L^-),
\qquad
(1+\beta_H)u'(H^+)=(1-\beta_H)u'(H^-).
\label{eq:transmission}
\end{equation}
\end{lemma}

\begin{proof}
Throughout the proof $z$ denotes either skew level, $z\in\{L,H\}$, and
$\beta_z$ the corresponding coefficient, so $\beta_z=\beta_L$ when $z=L$ and
$\beta_z=\beta_H$ when $z=H$. First, $u$ is a difference of convex functions
on a neighbourhood of each such $z$. Indeed $u'$ is $C^1$ on each side of $z$ up to $z$, with a finite
jump there, so its total variation near $z$ is bounded by $\int|u''|$ on the
two sides plus $|u'(z^+)-u'(z^-)|$, all finite. A continuous function whose
derivative has locally bounded variation is a difference of convex functions:
write $u'=g_1-g_2$ with $g_1,g_2$ non-decreasing and integrate. (Mere
boundedness of $u'$ near $z$ would not suffice, since a bounded derivative can
oscillate without bounded variation; this is why the one-sided $C^2$
hypothesis is imposed.) The It\^o--Tanaka--Meyer formula therefore applies: with $\ell^z$ the
symmetric local time and $u'$ read as the symmetric derivative
$\tfrac12(u'(z^+)+u'(z^-))$ on $\{S=z\}$,
\[
\begin{split}
u(S_t)&=u(S_0)+\int_0^t u'(S_s)\,dS_s
+\tfrac12\int_0^tu''(S_s)\,d\langle S\rangle_s\\
&\qquad+\tfrac12\sum_{z\in\{L,H\}}\bigl(u'(z^+)-u'(z^-)\bigr)\ell^z_t .
\end{split}
\]
We collect the coefficient of $d\ell^z_s$, which arises from two sources.
First, substituting \eqref{eq:SDE} into $\int u'(S_s)\,dS_s$: the local-time
part of $dS_s$ is $(\beta_z/z)S_s\,d\ell^z_s$, and since $\ell^z$ increases
only on $\{S_s=z\}$ this equals $\beta_z\,d\ell^z_s$; multiplied by the
symmetric derivative it contributes
$\tfrac12\beta_z\bigl(u'(z^+)+u'(z^-)\bigr)\,d\ell^z_s$. Second, the jump term
of the It\^o--Tanaka--Meyer formula contributes
$\tfrac12\bigl(u'(z^+)-u'(z^-)\bigr)\,d\ell^z_s$ directly. The total
coefficient of $d\ell^z_s$ is therefore
\begin{equation}
\begin{split}
\tfrac12\beta_z\bigl(u'(z^+)+u'(z^-)\bigr)
+\tfrac12\bigl(u'(z^+)-u'(z^-)\bigr)&=\tfrac12\Delta_z,\\
\Delta_z:=(1+\beta_z)u'(z^+)-(1-\beta_z)u'(z^-)&.
\end{split}
\label{eq:localtimecoef}
\end{equation}
Hence
$u(S_t)-\int_0^t\Lop u(S_s)\,ds
=\text{local martingale}+\tfrac12\sum_z\Delta_z\ell^z_t$.
If every $\Delta_z=0$ the right-hand side is a local martingale, so $u$ is in
the domain and the generator acts as $\Lop$. Conversely, $S$ is regular
(Proposition~\ref{prop:wellposed}), so from any starting point each $\ell^z$ is
a non-constant increasing additive functional; a non-zero $\tfrac12\Delta_z
\ell^z$ is then of finite variation and not identically zero, so it cannot be
a local martingale, and $u$ is in the domain only if $\Delta_z=0$. That the
resulting operator is the full generator, rather than a restriction of it, is
the single-skew statement of \citet{Lejay06} applied on each of the two levels
separately, which is legitimate because $\lambda<\eta$ and the two local times
have disjoint support; see also \citet{Ramirez11}.
\end{proof}

The terminology is from elliptic theory, where transmission conditions join
solutions across an interface of discontinuous coefficients; here the
interface is the skew point and the discontinuity is the local-time push.
Note what \eqref{eq:transmission} constrains: the function itself is
continuous at $L$ and $H$, and it is the \emph{slope} that jumps, in the
fixed ratio $u'(z^-)/u'(z^+)=(1+\beta_z)/(1-\beta_z)$ --- a quantity that will
be named the permeability of the level in \eqref{eq:permeability}. This should not be confused with the other derivative
condition in the paper. Smooth fit is a condition \emph{we impose}, at
the free boundaries $a$ and $b$: it makes $V'$ agree with $g'$ there, and it
is the optimality requirement that serves to determine those unknown
boundaries. Transmission, by contrast, is not a choice: by
Lemma~\ref{lem:transmission} every function in the domain of the generator
must satisfy it at the fixed levels $L$ and $H$. Whenever $\beta_z\ne0$ and
$V'(z)\ne0$ it makes $V'$ discontinuous at $z$; when $\beta_z=0$ the ratio is
$1$ and there is no kink. The counting works out: two unknowns $(a,b)$, and two
equations from smooth fit, once transmission has been used to glue the
branches together.

\paragraph{Terminology.} Let $u$ be continuous and piecewise $C^2$ as in
Lemma~\ref{lem:transmission}, and let $I\subset(0,\infty)$ be an open
interval. By the computation in that proof,
\[
d\bigl(e^{-rt}u(S_t)\bigr)=e^{-rt}(\Lop-r)u(S_t)\,dt
+\tfrac12e^{-rt}\sum_{z\in\{L,H\}}\Delta_z\,d\ell^z_t+dN_t
\]
with $N$ a local martingale and $\Delta_z$ as in \eqref{eq:localtimecoef}.
We call $u$ \emph{$r$-harmonic} on $I$ if $(\Lop-r)u=0$ on $I\setminus\{L,H\}$
and $\Delta_z=0$ at every skew point $z\in I$, that is, if $u$ solves
$(\Lop-r)u=0$ on $I$ in the sense of Lemma~\ref{lem:transmission}; then
$e^{-rt}u(S_t)$ is a local martingale while $S$ stays in $I$. We call $u$
\emph{$r$-superharmonic} on $I$ if instead $(\Lop-r)u\le0$ on
$I\setminus\{L,H\}$ and $\Delta_z\le0$ at every skew point in $I$; then
$e^{-rt}u(S_t)$ is a local supermartingale while $S$ stays in $I$. A
non-negative function $u$ is \emph{$r$-excessive} if $e^{-rt}u(S_t)$ is a
supermartingale under every $\Prob_x$, equivalently
$\E_x[e^{-r\tau}u(S_\tau)]\le u(x)$ for all $x$ and all stopping times $\tau$.
A non-negative $u$ that is $r$-superharmonic on all of $(0,\infty)$ is
$r$-excessive, since a non-negative local supermartingale is a
supermartingale. We say that $u$ is $r$-excessive at a point, or away from the
skew points, when the local conditions above hold there. Finally, the value
function $v$ is the least $r$-excessive majorant of $g$
\citep{DayanikKaratzas03}.

\begin{remark}[the exponential change of variable is free]
\label{rem:expfree}
For skew Brownian motion the condition at $\lambda$ is
$(1+\beta_L)\tilde u'(\lambda^+)=(1-\beta_L)\tilde u'(\lambda^-)$ with
$\tilde u(y)=u(e^y)$. Since $\tilde u'(y)=e^yu'(e^y)$ the factor $L$ appears on
both sides and cancels, so \eqref{eq:transmission} takes the same form in
price coordinates as in log coordinates.

The reason is that both one-sided derivatives are rescaled by the same
factor. If $S=f(X)$ for an increasing $C^1$ map $f$ and $z=f(\lambda)$, then
$u'(z^\pm)=\tilde u'(\lambda^\pm)/f'(\lambda)$, with one and the same
$f'(\lambda)$ on both sides because $f'$ is continuous. Condition
\eqref{eq:transmission} constrains only the ratio $u'(z^+)/u'(z^-)$, in which
this common factor cancels. The transmission condition is therefore a property
of the level itself, not of the coordinate in which the level is expressed.

That the coefficient is the \emph{same} $\beta$ in both coordinates is a
separate matter, and it is here that local time enters. Symmetric local time
rescales as $\ell^{f(\lambda)}(f(X))=f'(\lambda)\,\ell^\lambda(X)$, so with
$f=\exp$ one has $\ell^z(S)=z\,\ell^\lambda(X)$, and the local-time term
$(\beta_z/z)S\,d\ell^z(S)$ of \eqref{eq:SDE} equals $\beta_z z\,d\ell^\lambda(X)$
on $\{S=z\}$. This is exactly the term that It\^o's formula for $S=e^X$
produces from the local-time term of \eqref{eq:SDElog}, which is why
\eqref{eq:SDE} and \eqref{eq:SDElog} describe the same process with the same
$\beta_z$.
\end{remark}

\subsection{The stopping problem}

\begin{assumption}
\label{ass:drift}
$0<r$ and $\mu<r$.
\end{assumption}

Both halves are used, and for different purposes. Positivity of $r$ is what
makes the indicial roots straddle zero in Lemma~\ref{lem:exponents}, since the
product of the roots is $-r/(\tfrac12\sigma^2)$; the whole exponent structure,
and with it the fundamental pair of Definition~\ref{def:fund}, depends on it.
The inequality $\mu<r$ is what makes $(\Lop-r)g=(\mu-r)x$ strictly negative,
so that the pay-off is $r$-excessive away from the skew points.

Thus $e^{-rt}S_t$ is a supermartingale away from the skew points: absent the
local-time mechanism the holder would sell immediately, and every
continuation region below is generated by the skew alone.

With $g(x)=x$ the value function is
\begin{equation}
v(x):=\sup_{\tau\in\mathcal T}\E_x\!\left[e^{-r\tau}S_\tau\right],
\label{eq:value}
\end{equation}
where $\mathcal T$ is the set of stopping times of the filtration of $S$ and
$e^{-r\tau}S_\tau:=0$ on $\{\tau=\infty\}$.

\begin{proposition}[No equivalent martingale measure]
\label{prop:noEMM}
If $\beta_L\neq0$ or $\beta_H\neq0$ there is no $\Qmeas\sim\Prob$ under which
$(e^{-rt}S_t)_{t\ge0}$ is a local martingale.
\end{proposition}

\begin{proof}
Let $\Qmeas\sim\Prob$ with Girsanov kernel $\theta$, so
$dW_t=dW^\Qmeas_t-\theta_t\,dt$. Local time is a pathwise limit of
quadratic-variation functionals, hence invariant under an equivalent change of
measure, and the local-time terms of \eqref{eq:SDE} are unchanged. Then
\[
d\!\left(e^{-rt}S_t\right)=e^{-rt}S_t(\mu-\sigma\theta_t-r)\,dt
+e^{-rt}S_t\sigma\,dW^\Qmeas_t
+e^{-rt}S_t\Bigl(\tfrac{\beta_L}{L}d\ell^L_t+\tfrac{\beta_H}{H}d\ell^H_t\Bigr).
\]
The finite-variation part must vanish. The measures $d\ell^L,d\ell^H$ are
carried by $\{t:S_t=L\}$ and $\{t:S_t=H\}$, both Lebesgue-null, so they are
singular with respect to $dt$ and cannot cancel against the drift term; each
must vanish separately, forcing $\beta_L=\beta_H=0$.
\end{proof}

\section{Fundamental solutions and the crossing map}
\label{sec:prelim}

\subsection{Exponents}

Solutions $x^\alpha$ of $(\Lop-r)u=0$ require
\begin{equation}
\tfrac12\sigma^2\alpha(\alpha-1)+\mu\alpha-r=0,
\label{eq:indicial}
\end{equation}
with roots $\alpha_1<\alpha_2$.

\begin{lemma}
\label{lem:exponents}
Under Assumption~\ref{ass:drift} the roots of \eqref{eq:indicial} satisfy
$\alpha_1<0<1<\alpha_2$. Consequently
\begin{equation}
\kappa_1:=\frac{\alpha_2-1}{\alpha_2-\alpha_1},\qquad
\kappa_2:=\frac{1-\alpha_1}{\alpha_2-\alpha_1},\qquad
\kappa_1+\kappa_2=1,
\label{eq:kappa}
\end{equation}
are strictly positive.
\end{lemma}

\begin{proof}
The product of the roots is $-r/(\tfrac12\sigma^2)$, which is negative
because $r>0$; this gives $\alpha_1<0<\alpha_2$. Evaluating
\eqref{eq:indicial} at $\alpha=1$ gives $\mu-r<0$, so $\alpha_2>1$. Positivity of $\kappa_1,\kappa_2$ is immediate.
\end{proof}

On each of $(0,L)$, $(L,H)$, $(H,\infty)$ the general $r$-harmonic function is
$Ax^{\alpha_1}+Bx^{\alpha_2}$, and we identify it with its coefficient pair
$(A,B)^{\!\top}$. Coefficient pairs are columns throughout, so that the
crossing matrices of Lemma~\ref{lem:crossing} act on them from the left.
Evaluation is by the row vectors
\begin{equation}
e_x:=\bigl(x^{\alpha_1},x^{\alpha_2}\bigr),
\qquad
e_x':=\bigl(\alpha_1x^{\alpha_1},\alpha_2x^{\alpha_2}\bigr),
\qquad x>0,
\label{eq:evalvectors}
\end{equation}
which read off, from a pair $c$, the value $c\cdot e_x$ of the corresponding
function at $x$ and $x$ times its derivative $c\cdot e_x'$ there. These are
the only row vectors in the paper.

\subsection{The crossing matrix}

\begin{lemma}[Crossing]
\label{lem:crossing}
Let $\beta\in(-1,1)$ and $z>0$. Suppose $u=A^ux^{\alpha_1}+B^ux^{\alpha_2}$
above $z$ and $u=A^dx^{\alpha_1}+B^dx^{\alpha_2}$ below $z$, continuous at $z$
with $(1+\beta)u'(z^+)=(1-\beta)u'(z^-)$. Then
\begin{equation}
\begin{gathered}
\begin{pmatrix}A^d\\B^d\end{pmatrix}
=T_\beta(z)\begin{pmatrix}A^u\\B^u\end{pmatrix},\\[6pt]
T_\beta(z)=\frac{1}{(\alpha_2-\alpha_1)(1-\beta)}
\begin{pmatrix}
(\alpha_2-\alpha_1)-(\alpha_1+\alpha_2)\beta & -2\alpha_2\beta z^{\alpha_2-\alpha_1}\\[3pt]
2\alpha_1\beta z^{\alpha_1-\alpha_2} & (1+\beta)\alpha_2-(1-\beta)\alpha_1
\end{pmatrix}.
\end{gathered}
\label{eq:T}
\end{equation}
Moreover
\begin{equation}
\det T_\beta=\frac{1+\beta}{1-\beta},
\qquad
T_\beta^{-1}=T_{-\beta},
\qquad
T_0=I,
\label{eq:Tprops}
\end{equation}
where the argument $z$ is suppressed when it is clear from the context.
\end{lemma}

\begin{proof}
Write $a:=A^uz^{\alpha_1}$, $b:=B^uz^{\alpha_2}$, $\tilde a:=A^dz^{\alpha_1}$,
$\tilde b:=B^dz^{\alpha_2}$. Continuity reads $a+b=\tilde a+\tilde b$. Multiplying the
transmission condition by $z$ and using $z\,u'(z^\pm)=\alpha_1a+\alpha_2b$
respectively $\alpha_1\tilde a+\alpha_2\tilde b$,
\[
\alpha_1\tilde a+\alpha_2\tilde b=k(\alpha_1a+\alpha_2b),\qquad k:=\frac{1+\beta}{1-\beta}.
\]
Substituting $\tilde a=a+b-\tilde b$ gives
$(\alpha_2-\alpha_1)\tilde b=\alpha_1(k-1)a+(k\alpha_2-\alpha_1)b$, and with
$k-1=2\beta/(1-\beta)$ and
$k\alpha_2-\alpha_1=\bigl((1+\beta)\alpha_2-(1-\beta)\alpha_1\bigr)/(1-\beta)$
this is the second row of \eqref{eq:T} after dividing by $z^{\alpha_2}$. For
the first row,
\[
\tilde a=a\,\frac{(\alpha_2-\alpha_1)-\alpha_1(k-1)}{\alpha_2-\alpha_1}
 +b\,\frac{\alpha_2(1-k)}{\alpha_2-\alpha_1},
\]
and $(\alpha_2-\alpha_1)-\alpha_1(k-1)
=\bigl((\alpha_2-\alpha_1)-(\alpha_1+\alpha_2)\beta\bigr)/(1-\beta)$ while
$\alpha_2(1-k)=-2\alpha_2\beta/(1-\beta)$.

For \eqref{eq:Tprops}, expanding,
\begin{align*}
\det T_\beta
&=\frac{\bigl[(\alpha_2-\alpha_1)-(\alpha_1+\alpha_2)\beta\bigr]
        \bigl[(1+\beta)\alpha_2-(1-\beta)\alpha_1\bigr]
        +4\alpha_1\alpha_2\beta^2}
       {(\alpha_2-\alpha_1)^2(1-\beta)^2}\\
&=\frac{(\alpha_2-\alpha_1)^2(1-\beta^2)}{(\alpha_2-\alpha_1)^2(1-\beta)^2}
 =\frac{1+\beta}{1-\beta}.
\end{align*}
For $T_\beta^{-1}=T_{-\beta}$, note that the transmission condition
$(1+\beta)u'(z^+)=(1-\beta)u'(z^-)$ can be rewritten as
$(1+(-\beta))u'(z^-)=(1-(-\beta))u'(z^+)$: it is the condition with parameter
$-\beta$ and the roles of the two sides exchanged. The map from the
coefficients below $z$ to those above is therefore $T_{-\beta}$, and it
inverts $T_\beta$. (Direct multiplication confirms $T_\beta T_{-\beta}=I$.) Finally $T_0=I$ is immediate.
\end{proof}

The involution $T_\beta^{-1}=T_{-\beta}$ expresses the obvious symmetry:
crossing a skew point upwards with parameter $\beta$ is crossing it downwards
with $-\beta$. It halves the bookkeeping, and the determinant will turn out in
Theorem~\ref{thm:corner} to be exactly the quantity governing whether
resistance binds. We call
\begin{equation}
\pi(\beta):=\det T_\beta=\frac{1+\beta}{1-\beta}\in(0,\infty)
\label{eq:permeability}
\end{equation}
the \emph{permeability} of the level: $\pi>1$ for a support-type push,
$\pi<1$ for a resistance-type push, $\pi=1$ for no skew.

A solution of $(\Lop-r)u=0$ in the sense of Lemma~\ref{lem:transmission} is
determined on each of the three intervals $(0,L)$, $(L,H)$, $(H,\infty)$ by
its own coefficient pair; write these $c_0,c_1,c_2$ respectively. Because
$L<H$ the middle interval is non-degenerate and the two crossings act one at a
time, so Lemma~\ref{lem:crossing} gives $c_1=T_{\beta_H}(H)c_2$ and then
$c_0=T_{\beta_L}(L)c_1$. The linear map $c_2\mapsto c_0$ carrying the
behaviour above resistance to the behaviour below support is therefore
\begin{equation}
\begin{split}
c_0&=T_{\beta_L}(L)\,T_{\beta_H}(H)\,c_2,\\
\det\bigl(T_{\beta_L}T_{\beta_H}\bigr)
&=\pi(\beta_L)\pi(\beta_H)
=\frac{(1+\beta_L)(1+\beta_H)}{(1-\beta_L)(1-\beta_H)}.
\end{split}
\label{eq:compose}
\end{equation}
Every calculation below is an instance of reading \eqref{eq:compose}, or one
of its two factors, in one direction or the other.

\subsection{Smooth fit}

\begin{lemma}[Smooth fit]
\label{lem:smoothfit}
If $u(x)=Ax^{\alpha_1}+Bx^{\alpha_2}$ satisfies $u(m)=m$ and $u'(m)=1$ at some
$m>0$, then $(A,B)^{\!\top}=\vk(m)$ where
\begin{equation}
\vk(m):=\bigl(\kappa_1m^{1-\alpha_1},\ \kappa_2m^{1-\alpha_2}\bigr)^{\!\top}.
\label{eq:vkappa}
\end{equation}
\end{lemma}

\begin{proof}
With $a=Am^{\alpha_1}$, $b=Bm^{\alpha_2}$ the conditions read $a+b=m$ and
$\alpha_1a+\alpha_2b=m$, whence $b=\kappa_2m$, $a=\kappa_1m$.
\end{proof}

\subsection{The derivative identity}

\begin{lemma}[Derivative identity]
\label{lem:deriv}
With $\kappa_1,\kappa_2$ as in \eqref{eq:kappa} one has
$\kappa_1(1-\alpha_1)=-\kappa_2(1-\alpha_2)=:c$, where
$c=\frac{(\alpha_2-1)(1-\alpha_1)}{\alpha_2-\alpha_1}>0$, and hence
\begin{equation}
\frac{d\vk}{dm}(m)=c\,\bigl(m^{-\alpha_1},\,-m^{-\alpha_2}\bigr)^{\!\top}.
\label{eq:vkprime}
\end{equation}
\end{lemma}

\begin{proof}
Both products equal $(\alpha_2-1)(1-\alpha_1)/(\alpha_2-\alpha_1)$ up to sign,
and positivity follows from $\alpha_1<0<1<\alpha_2$.
\end{proof}

Identity \eqref{eq:vkprime} is what makes the free boundary tractable. The
weights $\kappa_1,\kappa_2$ do not vanish from derivatives but collapse into a
single constant: both components of $d\vk/dm$ carry the same factor $c$, with
opposite signs. Since $c>0$ it factors out of every derivative, never affects
a sign, and cancels in any ratio, so that what remains depends on the
exponents alone. This is why the sign arguments below reduce to comparing
$t^{\alpha_1}$ with $t^{\alpha_2}$.

\subsection{The fundamental pair and the one-zero property}

\begin{definition}[fundamental pair]
\label{def:fund}
Let $\psi_r$ be the solution of $(\Lop-r)u=0$ in the sense of
Lemma~\ref{lem:transmission} with $\psi_r=x^{\alpha_2}$ on $(0,L]$, extended
upward by $T_{-\beta_L}(L)$ and then $T_{-\beta_H}(H)$; and let $\varphi_r$ be
the solution with $\varphi_r=x^{\alpha_1}$ on $[H,\infty)$, extended downward
by $T_{\beta_H}(H)$ and then $T_{\beta_L}(L)$.
\end{definition}

Each of these is well defined, and the choices are forced up to a positive
scale. On $(0,L]$ the general solution is $Ax^{\alpha_1}+Bx^{\alpha_2}$; since
$\alpha_1<0<\alpha_2$, the first term blows up at $0$ and the second vanishes
there, so the solution vanishing at $0$ --- the increasing fundamental
solution --- has $A=0$, and $B=1$ merely fixes the scale. By
Lemma~\ref{lem:crossing} the coefficient pair on one side of a skew point
determines the pair on the other, so this choice on $(0,L]$ fixes $\psi_r$ on
all of $(0,\infty)$. Symmetrically, the solution vanishing at $\infty$ must be
a multiple of $x^{\alpha_1}$ on $[H,\infty)$, and its extension downward is
likewise unique. Thus $\psi_r$ and $\varphi_r$ are the two distinguished
elements of the two-dimensional solution space; they are linearly independent,
since their Wronskian never vanishes (proof of Lemma~\ref{lem:onezero} below),
so every solution is a combination $A\varphi_r+B\psi_r$.

\begin{lemma}[one zero]
\label{lem:onezero}
$\psi_r$ is positive and strictly increasing, $\varphi_r$ is positive and
strictly decreasing, and $\varphi_r/\psi_r$ is strictly decreasing on
$(0,\infty)$. Consequently, if $u\not\equiv0$ solves $(\Lop-r)u=0$ in the
sense of Lemma~\ref{lem:transmission} and $u(x_0)=0$ for some $x_0>0$, then
$x_0$ is the only zero of $u$ in $(0,\infty)$, and $u$ is strictly monotone.
\end{lemma}

\begin{proof}
\emph{Monotonicity on one piece.} Let $u=Ax^{\alpha_1}+Bx^{\alpha_2}$ satisfy
$u(z)=v>0$ and $u'(z)=s\ge0$ at some $z>0$. Solving for the coefficients,
\[
A=\frac{\alpha_2v-zs}{(\alpha_2-\alpha_1)z^{\alpha_1}},\qquad
B=\frac{zs-\alpha_1v}{(\alpha_2-\alpha_1)z^{\alpha_2}},
\]
so $B>0$ because $\alpha_1<0$. Now write
$u'(x)=x^{\alpha_1-1}\bigl(A\alpha_1+B\alpha_2x^{\alpha_2-\alpha_1}\bigr)$. The
bracket is strictly increasing in $x$, since $B\alpha_2>0$ and
$\alpha_2>\alpha_1$, and at $x=z$ it equals $z^{1-\alpha_1}s\ge0$. Hence
$u'>0$ on $(z,\infty)$, and $u>0$ there as well.

\emph{The function $\psi_r$.} On $(0,L]$, $\psi_r=x^{\alpha_2}$ has
$\psi_r>0$ and $\psi_r'>0$. At $L$ the transmission condition gives
$\psi_r'(L^+)=\psi_r'(L^-)/\pi(\beta_L)>0$, so the previous paragraph applies
at $z=L$ and gives $\psi_r>0$, $\psi_r'>0$ on $(L,H]$; the same argument at
$z=H$ extends this to $(H,\infty)$. Thus $\psi_r>0$ and $\psi_r'>0$
everywhere, one-sided derivatives being meant at $L$ and $H$.

\emph{The function $\varphi_r$.} The argument is the mirror image. If
$u(z)=v>0$ and $u'(z)=s<0$, then $A=(\alpha_2v-zs)/((\alpha_2-\alpha_1)
z^{\alpha_1})>0$, and writing
$u'(x)=x^{\alpha_2-1}\bigl(A\alpha_1x^{\alpha_1-\alpha_2}+B\alpha_2\bigr)$,
the bracket is strictly increasing in $x$ (because $A\alpha_1<0$ and
$x^{\alpha_1-\alpha_2}$ is decreasing) and equals $z^{1-\alpha_2}s<0$ at
$x=z$. Hence $u'<0$ on $(0,z)$, and $u>0$ there since $u$ increases as $x$
decreases from $z$. Starting from $\varphi_r=x^{\alpha_1}$ on $[H,\infty)$ and
using $\varphi_r'(z^-)=\pi(\beta_z)\varphi_r'(z^+)<0$ at $z=H$ and then $z=L$,
we obtain $\varphi_r>0$ and $\varphi_r'<0$ everywhere.

\emph{The Wronskian.} Let $\mathcal W:=\varphi_r\psi_r'-\varphi_r'\psi_r$. By
the two previous paragraphs, at every point
\[
\mathcal W=\underbrace{\varphi_r\,\psi_r'}_{>0}\;-\;
\underbrace{\varphi_r'\,\psi_r}_{<0}\;>\;0,
\]
so $\mathcal W$ cannot vanish anywhere. In particular $\varphi_r$ and
$\psi_r$ are linearly independent, and
$(\varphi_r/\psi_r)'=-\mathcal W/\psi_r^2<0$, so the ratio is strictly
decreasing.

\emph{One zero.} Every solution is therefore $u=A\varphi_r+B\psi_r$, and
$u(x_0)=0$ gives $B=-A\,\varphi_r(x_0)/\psi_r(x_0)$, so that
\[
u=\lambda\bigl(\psi_r(x_0)\varphi_r-\varphi_r(x_0)\psi_r\bigr),
\qquad \lambda:=A/\psi_r(x_0).
\]
Here $\psi_r(x_0)>0$, and $A\ne0$ because $A=0$ would force $B=0$ and hence
$u\equiv0$; so $\lambda\ne0$. Its zeros are the solutions of
$\varphi_r/\psi_r=\varphi_r(x_0)/\psi_r(x_0)$, and strict monotonicity of the
ratio leaves only $x_0$. Differentiating the same expression,
$u'=\lambda\bigl(\psi_r(x_0)\varphi_r'-\varphi_r(x_0)\psi_r'\bigr)$, and both
terms in the bracket are strictly negative, so $u'$ never vanishes either.
\end{proof}

\subsection{Sensitivity to the lower boundary}

For $0<a<L$ write $U_a$ for the solution of $(\Lop-r)u=0$ on $(a,\infty)$, in
the sense of Lemma~\ref{lem:transmission}, determined by the smooth fit
$U_a(a)=a$, $U_a'(a)=1$; by Lemma~\ref{lem:smoothfit} its coefficients on
$(a,L)$ are $\vk(a)$. Every candidate value function in this paper is $U_a$
restricted to its continuation band, so the following single lemma serves both
Section~\ref{sec:fb} and Section~\ref{sec:samesign}.

The restriction $a<L$ costs nothing, since every lower boundary arising in the
paper lies below support: in regime~A by \eqref{eq:aLb}, in regime~C by
Theorem~\ref{thm:corner}, along the regime-B branch because $a$ moves
monotonically between those two values, and in Section~\ref{sec:samesign}
by Proposition~\ref{lem:acomp}. It is in any case one of exposition only.
For any $a>0$ the argument below applies verbatim on the skew-free interval
immediately above $a$, where smooth fit again fixes the coefficients as
$\vk(a)$; so Lemma~\ref{lem:sensneg} holds for every $a>0$.

\begin{lemma}[the candidate falls as the lower boundary rises]
\label{lem:sensneg}
Let $w_a:=\partial U_a/\partial a$. Then $w_a(a)=0$,
$w_a'(a)=c(\alpha_1-\alpha_2)/a<0$, and
\[
w_a<0,\qquad w_a'<0 \qquad\text{throughout }(a,\infty).
\]
In particular $a\mapsto U_a(x)$ is strictly decreasing for each fixed $x>a$.
\end{lemma}

\begin{proof}
On $(a,L)$ the coefficients of $U_a$ are $\vk(a)$, so by
Lemma~\ref{lem:deriv} and \eqref{eq:vkprime},
\[
w_a(x)=c\,a^{-\alpha_1}x^{\alpha_1}-c\,a^{-\alpha_2}x^{\alpha_2}
=c\Bigl[(x/a)^{\alpha_1}-(x/a)^{\alpha_2}\Bigr],\qquad c>0,
\]
which vanishes at $x=a$, is negative for $x>a$ because
$\alpha_1<0<1<\alpha_2$, and has $w_a'(a)=c(\alpha_1-\alpha_2)/a<0$. The
crossing matrices $T_{\pm\beta}(z)$ do not depend on $a$, so differentiation
in $a$ commutes with them and $w_a$ is itself a solution in the sense of
Lemma~\ref{lem:transmission}. By Lemma~\ref{lem:onezero} its only zero is $a$
and $w_a'$ never vanishes, so both keep the signs they have just above $a$.
\end{proof}

\section{The free boundary}
\label{sec:fb}

We look for a continuation region of the form $(a,b)$ with
$0<a<L<b$, so that the stopping region is $(0,a]\cup[b,\infty)$; the ordering
$a<L<b$ is imposed here only to fix notation, and Theorem~\ref{thm:Airrelevant}
shows it is forced.

A word on terminology, since three different things could be called a
solution. A \emph{solution of the equation} is a function $u$ with
$(\Lop-r)u=0$ in the sense of Lemma~\ref{lem:transmission}. A \emph{solution
of the system} --- of \eqref{eq:master}, \eqref{eq:Asystem} or
\eqref{eq:Ceq}, as the regime dictates --- is a pair of numbers $(a,b)$
satisfying the smooth-fit and matching conditions; this is what ``the regime-A
solution'' and the like will mean below. A \emph{solution of the stopping
problem} is the value function $v$ together with an optimal rule. The first
two are algebra and the third is the object of interest; Section~\ref{sec:verify}
is what connects them, and Proposition~\ref{prop:selection} shows the second
does not by itself deliver the third. The candidate is
\begin{equation}
V(x)=\begin{cases}
x, & x\in(0,a]\cup[b,\infty),\\
\text{$r$-harmonic, glued by \eqref{eq:transmission}}, & x\in(a,b).
\end{cases}
\label{eq:candidate}
\end{equation}
That $L$ must lie inside the band, rather than being an assumption, is
confirmed by Theorem~\ref{thm:verification} and by every case in
Appendix~\ref{app:num}.

\begin{figure}[t]
\centering
\includegraphics[width=\textwidth]{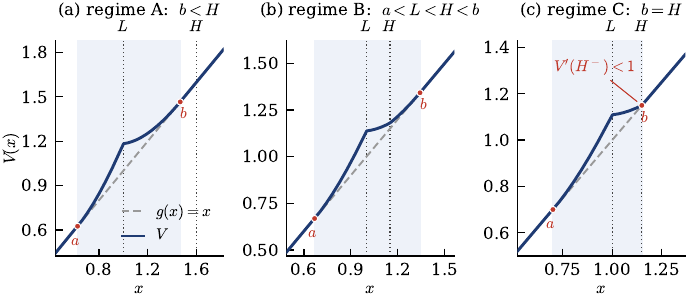}
\caption{The three geometries, with $r=0.05$, $\mu=0.01$, $\sigma^2=0.04$,
$L=1$ and $\beta_L=0.9$, so that the $L$-band is $(0.625,1.465)$. (a)
$H=1.60$ lies above the $L$-band: $b<H$, and the solution is that of the
single-skew problem at $L$ whatever the value of $\beta_H$ (here
$\beta_H=-0.30$). (b) $H=1.15$, $\beta_H=-0.20$: the band straddles $H$,
with smooth fit at both boundaries and a downward break of slope at each
skew point, by the factor $\pi(\beta_z)^{-1}$ of
\eqref{eq:transmission}. (c) $H=1.15$, $\beta_H=-0.50$, below
$\beta_H^\star=-0.392$: the upper boundary pins at $H$, where $V$ meets $g$
with slope $V'(H^-)<1$ and there is no smooth fit. Shaded: the continuation
region.}
\label{fig:regimes}
\end{figure}

Applying Lemma~\ref{lem:smoothfit} at each boundary and
Lemma~\ref{lem:crossing} at each skew point the band actually straddles gives
a single system in $(a,b)$:
\begin{equation}
T_{-\beta_L}(L)\,\vk(a)\;=\;T_{\beta_H}(H)\,\vk(b),
\label{eq:master}
\end{equation}
with the convention that a factor is replaced by the identity when the
corresponding skew point is not interior to $(a,b)$. Since $T_0=I$,
\eqref{eq:master} covers all cases at once. Three geometries are possible;
Fig.~\ref{fig:regimes} shows one instance of each.

\begin{description}
\item[Regime A] $b\le H$: the band lies below resistance; $T_{\beta_H}$ is
replaced by $I$.
\item[Regime B] $b>H$: the band straddles resistance; \eqref{eq:master} holds
as written.
\item[Regime C] $b=H$: the upper boundary pins at resistance; smooth fit fails
at $H$ and is replaced by an inequality.
\end{description}

Which regime obtains is decided by a single object. Consider the auxiliary
problem in which the skew at $H$ is switched off: formally $\beta_H=0$, which
is the same thing as deleting the level $H$ from the model. The two
descriptions agree at the level of the process, since with $\beta_H=0$ the
local-time term at $H$ drops out of \eqref{eq:SDE} altogether; and at the
level of the equation, since $T_0=I$ by \eqref{eq:Tprops}, so crossing $H$
leaves the coefficient pair unchanged. The auxiliary
problem thus has one skew point, and its continuation band is the pair
$(\bar a,\bar b)$ solving the two-equation system \eqref{eq:Asystem} recorded
in Theorem~\ref{thm:Airrelevant} below --- a system in which neither $H$ nor
$\beta_H$ appears. We call $(\bar a,\bar b)$ the \emph{$L$-band}. It is
defined whatever $H$ may be, is computable before anything about resistance is
known, and is in general \emph{not} the continuation band of the problem at
hand. Its role is to select the regime:
\begin{equation}
\text{regime A}\iff\bar b\le H,
\label{eq:regimeA}
\end{equation}
in which case the true band \emph{is} the $L$-band and
$\bar a<L<\bar b\le H$. The three cases are exhaustive. If $\bar b>H$, the
$L$-band contains $H$ but was computed as though the skew there were absent;
since $\beta_H<0$, the real process is pushed \emph{down} each time it touches
$H$, a penalty for a holder waiting to sell high which the $L$-band ignores,
so it is not the true band. Nor can the true upper boundary $b$ lie below $H$.
Suppose it did: a band that never reaches $H$ does not see $\beta_H$, so it
would be the $L$-band itself, giving $b=\bar b$; but $\bar b>H$, contradicting
$b<H$. The true upper boundary is thus at or above $H$, which is regime~B or
regime~C. The $L$-band
remains useful there, and indeed it is precisely when $\bar b>H$, so when
regime~A fails, that it reappears below as the endpoint of the regime-B
branch.

\subsection{Regime A: resistance is invisible}

\begin{theorem}
\label{thm:Airrelevant}
In regime A the pair $(a,b)$ and the value function are independent of
$\beta_H$, and $(a,b)$ solves
\begin{equation}
G(a)=G(b),
\qquad
(1+\beta_L)D(b)=(1-\beta_L)D(a),
\label{eq:Asystem}
\end{equation}
where
\[
\begin{split}
G(m)&:=\kappa_1L^{\alpha_1}m^{1-\alpha_1}+\kappa_2L^{\alpha_2}m^{1-\alpha_2},\\
D(m)&:=\kappa_1\alpha_1L^{\alpha_1}m^{1-\alpha_1}
      +\kappa_2\alpha_2L^{\alpha_2}m^{1-\alpha_2}.
\end{split}
\]
Moreover $G$ is strictly convex on $(0,\infty)$ with $G(0^+)=G(\infty)=\infty$
and its unique minimum at $m=L$, where $G(L)=L$. Consequently, for each
$v>L$ the level set $\{m>0:G(m)=v\}$ consists of exactly two points, one in
$(0,L)$ and one in $(L,\infty)$, and $G$ is injective on each side of $L$
separately. The first equation of \eqref{eq:Asystem} says that $a$ and $b$ lie
on one such level set, so any solution with $a<b$ satisfies
\begin{equation}
0<a<L<b
\label{eq:aLb}
\end{equation}
unless it is the degenerate $a=b=L$ of Corollary~\ref{cor:noband}. Support
therefore lies strictly inside the continuation band as a conclusion, not an
assumption.
\end{theorem}

\begin{proof}
When $b\le H$ the process stops before reaching $H$, so $\ell^H$ contributes
nothing on $[0,\tau]$ and $\beta_H$ enters neither \eqref{eq:master} nor the
candidate. The two displayed equations are the value-matching and transmission
conditions at $L$ written through Lemma~\ref{lem:smoothfit}: $G(m)$ is the
value at $L$ of the piece anchored at a smooth-fit boundary $m$, and $D(m)$ is
$L$ times its derivative there.

For the shape of $G$, write $t:=L/m$ and use
$\kappa_1(1-\alpha_1)=-\kappa_2(1-\alpha_2)=c>0$ from
Lemma~\ref{lem:deriv}. Differentiating,
\[
\frac{dG}{dm}=c\bigl(t^{\alpha_1}-t^{\alpha_2}\bigr),
\qquad
\frac{d^2G}{dm^2}=\frac{c}{m}
   \bigl(-\alpha_1t^{\alpha_1}+\alpha_2t^{\alpha_2}\bigr).
\]
Since $\alpha_1<0<1<\alpha_2$, both terms in $d^2G/dm^2$ are strictly
positive, so $G$ is strictly convex. In $dG/dm$, $t>1$ gives
$t^{\alpha_1}<1<t^{\alpha_2}$ and $t<1$ gives $t^{\alpha_1}>1>t^{\alpha_2}$,
so $dG/dm<0$ on $(0,L)$, $dG/dm>0$ on $(L,\infty)$ and $dG/dm=0$ at $m=L$:
the unique minimum is at $m=L$, with
$G(L)=(\kappa_1+\kappa_2)L=L$ by \eqref{eq:kappa}. The exponent $1-\alpha_2$
is negative and $1-\alpha_1$ positive, so $G(0^+)=G(\infty)=\infty$ and every
level $v>L$ is attained exactly twice, once on each side of $L$.

The ordering \eqref{eq:aLb} now needs nothing further. Since the derivative
of $G$ vanishes only at $L$, $G$ is strictly decreasing on the closed interval
$(0,L]$ and strictly increasing on $[L,\infty)$. Hence $G(a)=G(b)$ with $a<b$
is impossible when $a,b$ both lie in $(0,L]$, where it would force
$G(a)>G(b)$, and impossible when both lie in $[L,\infty)$, where it would
force $G(a)<G(b)$. Only $a<L<b$ and the degenerate $a=b=L$ remain.
\end{proof}

\begin{corollary}
\label{cor:noband}
If $\beta_L=0$ the system \eqref{eq:Asystem} forces $a=b=L$: the continuation
region is empty and the holder sells immediately. A continuation band exists
if and only if $\beta_L>0$.
\end{corollary}

Corollary~\ref{cor:noband} is the multi-skew form of a mechanism identified by
\citet{AlvarezSalminen17}, and it is worth spelling out. Start the process at
a skew point $z\in\{L,H\}$ and expand for small $t$. By It\^o--Tanaka,
\begin{equation}
\E_z\bigl[e^{-rt}g(S_t)\bigr]-g(z)
=\underbrace{\beta_z\,g'(z)\,\E_z\bigl[\ell^z_t\bigr]}_{\text{local-time push}}
+\underbrace{(\Lop-r)g(z)\,t}_{\text{drift and discounting}}+o(t),
\label{eq:shorttime}
\end{equation}
the two terms being the only contributions that survive to this order. The
point is that they are of \emph{different orders}: the process is at $z$ at
time $0$, so it accumulates local time there immediately, and
$\E_z[\ell^z_t]\asymp\sqrt t$ as $t\downarrow0$, whereas the drift term is
$O(t)$. Since $\sqrt t\gg t$ for small $t$, the first term decides the sign of
\eqref{eq:shorttime} however unfavourable the drift may be.

Hence $g$ fails to be $r$-excessive at $z$ exactly when $\beta_zg'(z)>0$:
there, holding for a short while beats stopping, so $z$ must lie in the
continuation region. With $g(x)=x$ and $\beta_L>0$ this places $L$ strictly
inside the band, which is Corollary~\ref{cor:noband}; at $H$, where
$\beta_H<0$, the push instead reinforces the drift and creates no value. Under
Assumption~\ref{ass:drift} the drift term is negative everywhere, so the
local-time push at $L$ is the only source of value in the problem.

\subsection{Regime C: selling exactly at resistance}

\begin{theorem}
\label{thm:corner}
Suppose the $L$-band has $\bar b>H$. Then the upper boundary pins at
$b=H$, with $a$ the unique root in $(0,L)$ of
\begin{equation}
\bigl[T_{-\beta_L}(L)\vk(a)\bigr]\cdot e_H=H,
\label{eq:Ceq}
\end{equation}
if and only if
\begin{equation}
\boxed{\ \pi(\beta_H)=\det T_{\beta_H}\ \le\ V'(H^-)\ }
\label{eq:criterion}
\end{equation}
where $V'(H^-)$ is computed from \eqref{eq:Ceq}. Equivalently, with
$\beta_H^\star:=\dfrac{V'(H^-)-1}{V'(H^-)+1}$, regime C obtains exactly
when $\beta_H\le\beta_H^\star$.
\end{theorem}

\begin{proof}
In regime C the point $H$ is a boundary of the continuation region, so
$V=g$ above it and $V'(H^+)=1$. The candidate is no longer in the domain of
the generator at $H$, so by \eqref{eq:localtimecoef} It\^o--Tanaka produces
there the extra term
$\tfrac12\bigl[(1+\beta_H)V'(H^+)-(1-\beta_H)V'(H^-)\bigr]\,d\ell^H_t$.
For $e^{-rt}V(S_t)$ to be a supermartingale this coefficient must be
non-positive, i.e.
$(1+\beta_H)\le(1-\beta_H)V'(H^-)$, which is \eqref{eq:criterion} after
dividing by $1-\beta_H>0$. Solving the equality for $\beta_H$ gives
$\beta_H^\star$. Uniqueness of the root of \eqref{eq:Ceq} follows because the
left-hand side is continuous and strictly monotone in $a$ on $(0,L)$.
\end{proof}

Theorem~\ref{thm:corner} is the paper's economic conclusion in one line. The
practitioner's rule ``sell at resistance'' is exactly optimal --- not
approximately, and not at some level near $H$ --- precisely when the
permeability of the resistance level falls below the slope of the value
function approaching it from inside. That the determinant of the crossing
matrix, introduced in Lemma~\ref{lem:crossing} as bookkeeping, is the quantity
that decides this is evidence that \eqref{eq:T} is the right parameterisation.

\subsection{Regime B and the transition}

In regime B equation \eqref{eq:master} holds with both factors present and
$\beta_H$ enters genuinely. The statement that the three geometries fit
together is Proposition~\ref{prop:transition} below; we first record the
Jacobian of the system. Using \eqref{eq:vkprime} for the derivatives of $\vk$,
and applying $\det[Pu\,|\,Qv]=\det P\cdot\det[u\,|\,P^{-1}Qv]$, valid for any
invertible $P$ and any $Q$, together with
$T_{-\beta}^{-1}=T_\beta$, the Jacobian of \eqref{eq:master} in $(a,b)$
satisfies
\begin{equation}
\det J=-c^2\,\det T_{-\beta_L}\;
\det\Bigl[\bigl(a^{-\alpha_1},-a^{-\alpha_2}\bigr)^{\!\top}\ \Big|\
T_{\beta_L}T_{\beta_H}\bigl(b^{-\alpha_1},-b^{-\alpha_2}\bigr)^{\!\top}\Bigr].
\label{eq:detJ}
\end{equation}

It is more convenient, however, to reduce the system to a scalar one. The
lower boundary $a$ fixes the coefficient pair below $L$ by smooth fit, and
crossing $L$ carries it to the middle interval; write
\begin{equation}
\Phi(a):=T_{-\beta_L}(L)\,\vk(a)
\label{eq:Phi}
\end{equation}
for the resulting coefficient pair on $(L,H)$, so that $V(x)=\Phi(a)\cdot e_x$
there. Above $H$ no crossing is involved: since
$(H,b)$ contains no skew point, smooth fit at $b$ alone fixes the coefficient
pair there, and it is $\vk(b)$. The four quantities we need are the
evaluation vectors $e_H,e_H'$ of \eqref{eq:evalvectors}, taken at $x=H$,
applied to the coefficient pairs on either side of $H$:
\begin{equation}
\begin{aligned}
A(a)&:=\Phi(a)\cdot e_H=V(H^-),
&\qquad
B(a)&:=\Phi(a)\cdot e_H'=HV'(H^-),\\
G_H(b)&:=\vk(b)\cdot e_H=V(H^+),
&\qquad
D_H(b)&:=\vk(b)\cdot e_H'=HV'(H^+).
\end{aligned}
\label{eq:ABGD}
\end{equation}
Written out,
\[
\begin{split}
G_H(b)&=\kappa_1H^{\alpha_1}b^{1-\alpha_1}
+\kappa_2H^{\alpha_2}b^{1-\alpha_2},\\
D_H(b)&=\kappa_1\alpha_1H^{\alpha_1}b^{1-\alpha_1}
+\kappa_2\alpha_2H^{\alpha_2}b^{1-\alpha_2}.
\end{split}
\]
\begin{corollary}[Sensitivity to the lower boundary]
\label{lem:sens}
With $w_a=\partial U_a/\partial a$ as in Lemma~\ref{lem:sensneg},
\[
\frac{dA}{da}(a)=w_a(H)<0,
\qquad
\frac{dB}{da}(a)=H\,w_a'(H)<0 .
\]
\end{corollary}

\begin{proof}
$V=U_a$ on $(a,b)$, and by \eqref{eq:ABGD} the quantities $A$ and $B$ are the
value and scaled derivative of $V$ at $H^-$; differentiating in $a$ replaces
$V$ by $w_a$, and Lemma~\ref{lem:sensneg} signs both.
\end{proof}

Continuity and transmission at $H$ read
\[
A(a)=G_H(b),
\qquad
B(a)=\pi(\beta_H)\,D_H(b).
\]
By Corollary~\ref{lem:sens}, $A$ is strictly decreasing and hence injective,
so the first equation determines $a$ from $b$ uniquely: write
$a(b):=A^{-1}\bigl(G_H(b)\bigr)$, defined for those $b$ with $G_H(b)$ in the
range of $A$. Substituting into the second equation,
\begin{equation}
\pi(\beta_H)=\frac{B\bigl(a(b)\bigr)}{D_H(b)}.
\label{eq:scalar}
\end{equation}
The two-unknown system has therefore collapsed to a single scalar. Choosing
$b$ fixes $a$, and then \eqref{eq:scalar} fixes $\beta_H$; so for each
admissible $b$ there is exactly one pair $(a,\beta_H)$, and the regime-B
branch is the curve
\[
b\ \longmapsto\ \bigl(a(b),\,b,\,\beta_H(b)\bigr).
\]
This is what we mean by a graph over $b$: the branch is parameterised by $b$
alone and cannot double back on itself. Proposition~\ref{prop:transition}
turns that into the statement that the optimal threshold does not jump.

\begin{lemma}[the candidate is positive and increasing]
\label{lem:incr}
Let $(a,b)$ with $a<L<b$ solve the relevant instance of \eqref{eq:master} or
\eqref{eq:Ceq}, and let $V$ be the candidate \eqref{eq:candidate}. Then $V>0$ and $V'>0$
throughout $(a,b)$.
\end{lemma}

\begin{proof}
For positivity, note that $V$ on $(a,b)$ is the restriction of a solution
$u=c_1\psi_r+c_2\varphi_r\not\equiv0$ of $(\Lop-r)u=0$ on $(0,\infty)$ in the
sense of Lemma~\ref{lem:transmission}, with $u(a)=a>0$ and $u(b)=b>0$. If
$V(x_0)=0$ for some $x_0\in(a,b)$, then by Lemma~\ref{lem:onezero} $u$ is
strictly monotone with $x_0$ as its only zero, so $u(a)$ and $u(b)$ have
opposite signs --- a contradiction. Hence $V>0$. (The skew points need no
separate treatment: Lemma~\ref{lem:onezero} already accounts for them.)

For monotonicity, smooth fit gives $V'(a^+)=1>0$. Suppose $V'$ first vanished
at some $x_1$ in a skew-free subinterval. Approaching $x_1$ from the left with
$V'>0$ and $V'(x_1)=0$ forces $V''(x_1)\le0$; but
$\tfrac12\sigma^2x_1^2V''(x_1)=rV(x_1)>0$ by the previous paragraph, so
$V''(x_1)>0$, a contradiction. The same argument applies one-sidedly if $x_1=z\in\{L,H\}$ with
$V'(z^-)=0$, since $\tfrac12\sigma^2z^2V''(z^-)=rV(z)>0$; and otherwise the
transmission condition gives $V'(z^+)=V'(z^-)/\pi(\beta_z)$ with
$\pi(\beta_z)>0$, which preserves the sign. So $V'>0$ on all of $(a,b)$.
\end{proof}

\begin{proposition}
\label{prop:transition}
Fix $\beta_L,L,H$ with $\bar b>H$, so that regime~A fails by
\eqref{eq:regimeA}, and let $(\bar a,\bar b)$ be the $L$-band. Then
$\bar b$ is the largest upper boundary attainable in regime~B, whence the
notation. Along the regime-B branch,
$\pi(\beta_H)$ is a strictly increasing function of $b$, running from
$V'(H^-)$ at $b=H$ to $1$ at $b=\bar b$; in particular $V'(H^-)<1$. Consequently $b\mapsto\beta_H$ is a strictly
increasing bijection from $(H,\bar b)$ onto $(\beta_H^\star,0)$, the Jacobian
\eqref{eq:detJ} is non-singular throughout, and the transition to regime C at
$\beta_H=\beta_H^\star$ is a transversal crossing of $\{b=H\}$ rather than a
fold: the optimal threshold does not jump.
\end{proposition}

\begin{proof}
See Appendix~\ref{app:proofs}.
\end{proof}

\begin{remark}[what transversality rules out]
\label{rem:nofold}
Write the regime-B branch as the curve
\[
b\ \longmapsto\ \bigl(a(b),\,\beta_H(b)\bigr),\qquad b\in(H,\bar b).
\]
Regime~B is valid where $b>H$, so $\{b=H\}$ bounds its region of validity. Two things could in principle go wrong there, and
Proposition~\ref{prop:transition} excludes both.

The first is a \emph{fold}. If $\beta_H(b)$ had an interior turning point
$b^\ast$, with $d\beta_H/db=0$ there, then $\beta_H$ would fail to be monotone
along the branch: values of $\beta_H$ just inside the turn would admit
\emph{two} regime-B solutions and values just outside would admit
\emph{none}. As $\beta_H$ swept past the fold the optimal $b$ would have to
jump from one branch to the other, and the model would display hysteresis in a
parameter --- the same resistance strength supporting two different optimal
rules, with the selection depending on where one came from. Strict
monotonicity of $\pi(\beta_H)$ in $b$ is exactly the statement that no such
turning point exists.

The second is a \emph{tangential} arrival: the branch could meet the line
$\{b=H\}$ tangentially in the $(b,\beta_H)$-plane, that is with
$db/d\beta_H\to0$, equivalently $d\beta_H/db\to\infty$, as $b\downarrow H$.
Then $b$ would leave $H$ with zero speed as $\beta_H$ rises past
$\beta_H^\star$, so to first order it would stay pinned at $H$ as in regime~C:
the hand-over from C to B would be degenerate, in the sense that the two
regimes could not be told apart at first order near $\beta_H^\star$, and
recovering $\beta_H$ from $b$ would be infinitely ill-conditioned there.
\emph{Transversal} crossing means the branch meets $\{b=H\}$ at a non-zero
angle: by \eqref{eq:transversal}, established in Step~6 of the proof in
Appendix~\ref{app:proofs}, $d\beta_H/db$ has a finite, positive limit
as $b\downarrow H$ (finite rules out tangency; positive is the first-order
form of the no-fold property above). With $\beta_L=0.90$, $L=1$ the limit
\eqref{eq:transversal} equals $0.752$, $0.668$ and $0.522$ at $H=1.30$,
$1.10$ and $1.05$ respectively, matching finite differences along the
branch.

The consequence is the one stated: $a$ and $b$ are continuous in $\beta_H$
across $\beta_H^\ast$. Below it, $b$ pins at $H$ and $a$ continues from the
value it had on the branch. This is also what makes $b\mapsto\beta_H$
invertible, which Theorem~\ref{thm:ident} turns into the identification
statement.
\end{remark}

\section{Verification}
\label{sec:verify}

The hysteresis models require the reader to check that the candidate dominates
the pay-off, and there the condition is not automatic. Here it is a theorem, and the reason is structural: in regimes A and B
there is smooth fit at \emph{both} boundaries, and in regime C the missing
smooth fit at $b=H$ is replaced by the slope bound $V'(H^-)<1$ of
Proposition~\ref{prop:transition}.

\begin{proposition}[Dominance]
\label{prop:dominance}
Let $W:=V-g$ on $(a,b)$, with $V$ the candidate \eqref{eq:candidate}, $g(x)=x$
and $L\in(a,b)$. Under Assumption~\ref{ass:drift}, $W>0$ on $(a,b)$.
\end{proposition}

\begin{proof}
On any subinterval of $(a,b)$ free of skew points, $(\Lop-r)V=0$ and
$(\Lop-r)g=(\mu-r)x$, so
\begin{equation}
\tfrac12\sigma^2x^2W''+\mu xW'-rW=(r-\mu)x>0.
\label{eq:Wode}
\end{equation}

\emph{Boundaries.} Smooth fit gives $W(a)=W'(a)=0$, and \eqref{eq:Wode} at
$x=a$ gives $\tfrac12\sigma^2a^2W''(a)=(r-\mu)a>0$; hence $W''(a)>0$ and
$W>0$ immediately above $a$. In regimes A and B, identically, $W(b)=W'(b)=0$,
$W''(b)>0$ and $W>0$ immediately below $b$. In regime C, $b=H$ is not a free
boundary and there is no smooth fit there; instead value matching gives
$W(H)=0$, and Proposition~\ref{prop:transition} (whose hypothesis
$\bar b>H$ is that of Theorem~\ref{thm:corner}) gives $V'(H^-)<1$, that is
$W'(H^-)<0$. So again $W>0$ and $W'<0$ immediately below $b$.

\emph{No interior non-negative critical point.} Let $x_1$ be an interior
critical point of $W$ in a skew-free subinterval with $W'(x_1)=0$ and
$W''(x_1)\le0$. Then \eqref{eq:Wode} gives
\[
rW(x_1)=\tfrac12\sigma^2x_1^2W''(x_1)-(r-\mu)x_1\le-(r-\mu)x_1<0,
\]
so $W(x_1)<0$. Hence $W$ admits no interior local maximum at a point where it
is non-negative.

\emph{Propagation on $(a,L)$.} We have $W>0$ and $W'>0$ just above $a$.
Suppose $W'$ first vanishes at $x_2\in(a,L)$; then $W(x_2)>0$, and since $W'$
is decreasing into $x_2$ we have $W''(x_2)\le0$, giving a non-negative
interior local maximum, a contradiction. So $W'>0$ throughout $(a,L)$ and in
particular $W(L)>0$.

\emph{Propagation on $(L,b)$.} Just below $b$ we have $W>0$ and $W'<0$ by
the boundary step: $W$ decreases to $0$ as $x\uparrow b$. Now
run the argument leftward from $b$. Suppose $W'$ first vanishes, moving left,
at $x_2\in(L,b)$, so $W'<0$ on $(x_2,b)$. Then $W$ is strictly decreasing on $[x_2,b]$ with $W(b)=0$, so $W(x_2)>0$;
and since
$W'$ decreases from $0$ as $x$ increases past $x_2$, $W''(x_2)\le0$ --- again
a non-negative interior local maximum, a contradiction. So $W'<0$ throughout
$(L,b)$, and $W>0$ on $[L,b)$. If $H\in(a,b)$ (regime B) the argument is applied separately on
$(H,b)$ and $(L,H)$; the two are joined by
$V'(H^-)=\pi(\beta_H)V'(H^+)$ with $\pi(\beta_H)\in(0,1)$, which preserves the
sign of $V'$ and keeps $V'\in(0,1)$, so $W'<0$ continues across the kink and
$W(H)>0$ is inherited.

The change of sign of $W'$ at $L$ is carried by a jump in $W'$ itself, not by
$W''$: by \eqref{eq:transmission}, $W'(L^+)+1=\bigl(W'(L^-)+1\bigr)/\pi(\beta_L)$
with $\pi(\beta_L)>1$, and the two propagation steps show
$W'(L^-)\ge0\ge W'(L^+)$. So $W$ has a concave corner at $L$, where it attains
its maximum over $(a,b)$; $W''$ is undefined there, while near $a$ and $b$ it is
positive.
\end{proof}

\begin{theorem}[Verification]
\label{thm:verification}
Let $(a,b)$ solve the relevant instance of \eqref{eq:master} or
\eqref{eq:Ceq}, and let $V$ be the candidate \eqref{eq:candidate}. Then
$V=v$ and $\tau^\ast=\inf\{t:S_t\notin(a,b)\}$ is optimal.
\end{theorem}

\begin{proof}
The candidate $V$ is $C^2$ on each closed piece between consecutive points of
$\{a,b,L,H\}$ --- it is $g$ outside $(a,b)$ and a combination of
$x^{\alpha_1},x^{\alpha_2}$ on each skew-free subinterval inside --- so it is
a difference of convex functions near each of these points exactly as in the
proof of Lemma~\ref{lem:transmission}, and has bounded derivative on compacts.
It\^o--Tanaka and the computation \eqref{eq:localtimecoef} therefore give
\[
d\!\left(e^{-rt}V(S_t)\right)=e^{-rt}(\Lop-r)V(S_t)\,dt
+\tfrac12 e^{-rt}\!\!\sum_{z\in\{L,H\}}\!\!\Delta_z\,d\ell^z_t+dN_t
\]
for a local martingale $N$, with
$\Delta_z=(1+\beta_z)V'(z^+)-(1-\beta_z)V'(z^-)$ as there. We check that both non-martingale terms are
non-positive.

The first term is a $dt$-integral, so it only sees $(\Lop-r)V$ off a
Lebesgue-null set of times: $S$ spends zero time at any single level (its
occupation measure is absolutely continuous), so the finitely many points
$a,b,L,H$, where $V''$ has only one-sided values, do not contribute. On
$(a,b)\setminus\{L,H\}$ the integrand is identically zero, because $V$ solves
$(\Lop-r)V=0$ there by construction. On the stopping region
$V=g$, so $(\Lop-r)V=(\mu-r)x<0$ by Assumption~\ref{ass:drift}. Hence the
first term is non-positive, and strictly negative only while $S$ is outside
$[a,b]$.

For the local-time terms: $L\in(a,b)$ always, so $\Delta_L=0$ by
\eqref{eq:transmission}. For $H$: in regime B, $H\in(a,b)$ and $\Delta_H=0$;
in regime A, $H$ lies in the stopping region where $V=g$, so
$\Delta_H=(1+\beta_H)-(1-\beta_H)=2\beta_H<0$ by \eqref{eq:betasigns}; in
regime C, $\Delta_H\le0$ is exactly the criterion \eqref{eq:criterion}
defining the regime.

Hence $e^{-rt}V(S_t)$ is a supermartingale. Optional sampling and
Proposition~\ref{prop:dominance} give $V\ge v$, and evaluating along
$\tau^\ast$, where all the inequalities are equalities, gives $V=v$.
\end{proof}

\begin{remark}
\label{rem:noresidual}
Theorem~\ref{thm:verification} leaves nothing for the reader to check. This
should be contrasted with problems in which an exit level is fixed by the
model rather than chosen, where the analogue of
Proposition~\ref{prop:dominance} is a hypothesis that has to be checked case
by case and need not hold. The difference lies at the boundaries. At a free boundary, one chosen
optimally, smooth fit supplies $W=W'=0$, which \eqref{eq:Wode} converts into
local positivity; at an exit level fixed by the model rather than chosen, only
value matching $W=0$ is available, and the sign of $W'$ there is not
controlled. In regimes A and B both $a$ and $b$ are free. In regime C the upper
boundary is fixed at $H$, and it is the slope bound $V'(H^-)<1$ that takes the
place of smooth fit.
\end{remark}

\begin{remark}[the sign convention is load-bearing]
\label{rem:signload}
The proof of Theorem~\ref{thm:verification} uses $\beta_H<0$ exactly once, in
regime A, where $\Delta_H=2\beta_H$. Were $\beta_H>0$, regime A would be
inadmissible: a positive skew sitting in the stopping region always breaks the
candidate, since waiting near such a point is valuable. The correct geometry
would then be \emph{two} bands, one around each skew point, with a
three-component stopping region; this is carried out in
Section~\ref{sec:samesign}. The assumption that resistance pushes downwards is
therefore what collapses the problem to a single band.
\end{remark}

\section{Identification}
\label{sec:ident}

\begin{theorem}
\label{thm:ident}
Fix $r,\mu,\sigma,L,H,\beta_L$ and let $\beta_H^\star$ be as in
Theorem~\ref{thm:corner}. Then the optimal rule --- equivalently the pair
$(a,b)$, equivalently the value function --- determines $\beta_H$ if and only
if
\begin{equation}
\beta_H\in(\beta_H^\star,\,0)\quad\text{and the regime-A band satisfies }b>H.
\label{eq:identset}
\end{equation}
Outside this set $\beta_H$ is not identified: in regime A the rule is exactly
independent of $\beta_H$, and in regime C every $\beta_H\le\beta_H^\star$
produces the same rule $b=H$.
\end{theorem}

\begin{proof}
Regime A: Theorem~\ref{thm:Airrelevant}. Regime C:
Theorem~\ref{thm:corner} determines $b=H$ irrespective of the value of
$\beta_H$ below $\beta_H^\star$. Regime B: the map
$\beta_H\mapsto b$ is strictly monotone along the branch by
Proposition~\ref{prop:transition}, hence injective.
\end{proof}

The content of Theorem~\ref{thm:ident} is a limit on what liquidation data can
reveal. In regime C one can recover the fact that the resistance level binds,
and its location, since $b=H$; but not its strength beyond the one-sided
bound $\beta_H\le\beta_H^\star$. A mildly sticky resistance level and an impenetrable one
generate identical behaviour, and the boundary between the identified and
unidentified cases is explicit.

Support behaves differently. In regime A the resistance parameters drop out
entirely (Theorem~\ref{thm:Airrelevant}), and the band then determines the
support parameters uniquely and in closed form.

\begin{theorem}[Inversion]
\label{thm:inversion}
Fix $r>\mu$ and $\sigma>0$ and let $0<a<b$. Define
\begin{equation}
L:=\left[\frac{\kappa_1\bigl(b^{1-\alpha_1}-a^{1-\alpha_1}\bigr)}
              {\kappa_2\bigl(a^{1-\alpha_2}-b^{1-\alpha_2}\bigr)}
   \right]^{1/(\alpha_2-\alpha_1)},
\qquad
\beta_L:=\frac{D(a)-D(b)}{D(a)+D(b)},
\label{eq:inversion}
\end{equation}
with $D$ as in Theorem~\ref{thm:Airrelevant}. Both quantities are well
defined, since $1-\alpha_1>0>1-\alpha_2$ makes both brackets strictly
positive.
\begin{enumerate}
\item[(i)] If $(a,b)$ is the regime-A band of the model \eqref{eq:SDE} with
some support parameters $(L',\beta_L')$, then $(L',\beta_L')=(L,\beta_L)$.
\item[(ii)] Conversely, if $\beta_L\in(0,1)$ and $a<L<b$, then the model
\eqref{eq:SDE} with support parameters $(L,\beta_L)$, any resistance level
$H\ge b$ and any $\beta_H\in(-1,0)$ has continuation region exactly $(a,b)$,
and its optimal rule is the first exit of $(a,b)$.
\end{enumerate}
\end{theorem}

\begin{proof}
A regime-A band satisfies \eqref{eq:Asystem}. Its first equation is
$G(a)=G(b)$, i.e.
$\kappa_1L^{\alpha_1}(a^{1-\alpha_1}-b^{1-\alpha_1})
+\kappa_2L^{\alpha_2}(a^{1-\alpha_2}-b^{1-\alpha_2})=0$; dividing by
$L^{\alpha_1}$ leaves a single positive value of $L^{\alpha_2-\alpha_1}$, the
one displayed. With $L$ fixed, the second equation
$(1+\beta_L)D(b)=(1-\beta_L)D(a)$ is linear in $\beta_L$ with the stated
unique solution. This proves (i). For (ii), the same computation read
backwards shows that $(a,b)$ solves \eqref{eq:Asystem} for $(L,\beta_L)$.
Since $H\ge b$ the band does not reach $H$, so regime A applies and $\beta_H$
is irrelevant by Theorem~\ref{thm:Airrelevant}. Optimality is
Theorem~\ref{thm:verification}, whose hypotheses hold because smooth fit
obtains at both boundaries.
\end{proof}

Together, Theorems~\ref{thm:ident} and~\ref{thm:inversion} describe what
liquidation behaviour reveals about the two levels. Support, whose push is
what creates the continuation band in the first place, is recovered exactly
from a regime-A band. Resistance is recovered only when the band straddles it,
and when the band pins at $H$ only its location and the bound
$\beta_H\le\beta_H^\star$ are revealed.

\section{Two support levels: the same-sign case}
\label{sec:samesign}

Remark~\ref{rem:signload} observed that the assumption $\beta_H<0$ is used
exactly once in the verification, and that dropping it must change the
geometry. We now carry that case out. Throughout this section
\begin{equation}
\beta_L\in(0,1),\qquad \beta_H\in(0,1),
\label{eq:samesign}
\end{equation}
so that $H$ is a second support level rather than a resistance level. Both
$L$ and $H$ then satisfy $\beta g'>0$, so by the mechanism of
Corollary~\ref{cor:noband} both lie in the continuation region, and the
question is whether they share one band or occupy two.

\subsection{Decoupling}

Write $(a_1,b_1)$ for the band produced by Theorem~\ref{thm:Airrelevant} with
skew point $L$ and parameter $\beta_L$, and $(a_2,b_2)$ for the band produced
by the same theorem with $H$ and $\beta_H$. Each solves a single-skew system
involving only its own level.

\begin{proposition}[Decoupling]
\label{prop:decouple}
Assume \eqref{eq:samesign} and $b_1\le a_2$. Then the continuation region is
$(a_1,b_1)\cup(a_2,b_2)$, the value function is the corresponding candidate,
and the stopping region
\[
(0,a_1]\ \cup\ [b_1,a_2]\ \cup\ [b_2,\infty)
\]
has three components. In particular $(a_1,b_1)$ does not depend on
$(H,\beta_H)$ and $(a_2,b_2)$ does not depend on $(L,\beta_L)$.
\end{proposition}

Here $(a_1,b_1)$ is the $L$-band $(\bar a,\bar b)$ of
Section~\ref{sec:fb}, and $(a_2,b_2)$ is its counterpart at $H$, obtained from
\eqref{eq:Asystem} with $(L,\beta_L)$ replaced by $(H,\beta_H)$. In
Section~\ref{sec:fb} the $L$-band was an auxiliary object: it is computed with
the skew at $H$ switched off, and served only to decide which regime applies.
Under \eqref{eq:samesign} and $b_1\le a_2$ it is no longer auxiliary: the
$L$-band and its counterpart at $H$ are themselves the two components of the
continuation region.

\begin{proof}
Write $V_L$ for the candidate of Theorem~\ref{thm:Airrelevant} with skew point
$L$ --- equal to the band value on $(a_1,b_1)$ and to $g$ elsewhere --- and
$V_H$ likewise for $H$, and set
\begin{equation}
\widetilde V:=V_L+V_H-g.
\label{eq:additivecap}
\end{equation}
We show $\widetilde V$ is an $r$-excessive majorant of $g$ and coincides with the
candidate; optimality then follows from the characterisation of $v$ as the
least such majorant.

\emph{Majorant.} $V_L\ge g$ and $V_H\ge g$ by
Proposition~\ref{prop:dominance}, applied to the single-skew problem at $L$
and, with $(L,\beta_L)$ replaced by $(H,\beta_H)$, at $H$; so $\widetilde V-g=(V_L-g)+(V_H-g)\ge0$.

\emph{Local-time defects.} Since $b_1\le a_2$ we have $L<a_2$ and $H>b_1$, so
$V_H=g$ on a neighbourhood of $L$ and $V_L=g$ on a neighbourhood of $H$.
Writing $\Delta_z$ for the defect of Theorem~\ref{thm:verification} at $z$ and
using $\Delta_z(g)=2\beta_z$,
\[
\Delta_L(\widetilde V)=\underbrace{\Delta_L(V_L)}_{=0}
           +\underbrace{\Delta_L(V_H)}_{=\Delta_L(g)=2\beta_L}
           -\;2\beta_L=0,
\]
because $V_L$ satisfies the transmission condition at $L$, that point being
interior to its own band; and symmetrically $\Delta_H(\widetilde V)=0$. Both defects
vanish identically, whatever the signs of $\beta_L,\beta_H$.

\emph{Drift.} Off the skew points $(\Lop-r)V_L$ vanishes on $(a_1,b_1)$ and
equals $(\mu-r)x$ elsewhere, and similarly for $V_H$. Hence
\[
(\Lop-r)\widetilde V=\begin{cases}
0, & x\in(a_1,b_1)\cup(a_2,b_2),\\
(\mu-r)x<0, & x\notin[a_1,b_1]\cup[a_2,b_2],
\end{cases}
\]
using Assumption~\ref{ass:drift} in the second line. The two bands are
disjoint by hypothesis, so these are the only cases and $(\Lop-r)\widetilde V\le0$
throughout. Therefore $\widetilde V$ is $r$-excessive.

\emph{Identification.} On $(a_1,b_1)$ we have $V_H=g$ and so $\widetilde V=V_L$; on
$(a_2,b_2)$, $\widetilde V=V_H$; and on the complement $\widetilde V=g$. Thus $\widetilde V$ is exactly the
candidate of the statement.

\emph{Conclusion.} Here $v$ is the value function \eqref{eq:value}. First,
$v\le \widetilde V$: for any $x$ and $\tau\in\mathcal T$, since $\widetilde V\ge g$ and
$e^{-rt}\widetilde V(S_t)$ is a non-negative supermartingale, optional sampling gives
$\E_x[e^{-r\tau}g(S_\tau)]\le\E_x[e^{-r\tau}\widetilde V(S_\tau)]\le \widetilde V(x)$, with both
integrands read as $0$ on $\{\tau=\infty\}$; taking the supremum over $\tau$
gives $v(x)\le \widetilde V(x)$. (This is the fact that $v$ is the least $r$-excessive
majorant of $g$, used in one direction.) Second,
$v\ge \widetilde V$: let $\tau$ be the first exit time of $(a_1,b_1)\cup(a_2,b_2)$. If
$x$ lies outside both bands then $\tau=0$ and $\widetilde V(x)=g(x)$. If
$x\in(a_i,b_i)$, then $\tau$ is the exit time of that bounded interval, finite
almost surely since $S$ is regular; up to $\tau$ the process
$e^{-rt}\widetilde V(S_t)$ is a local martingale, because $\widetilde V$ is $r$-harmonic on the
band, and it is bounded, because $\widetilde V$ is continuous on $[a_i,b_i]$. Optional
stopping and $\widetilde V=g$ at $a_i$ and $b_i$ give
$\widetilde V(x)=\E_x[e^{-r\tau}g(S_\tau)]$, and this is at most $v(x)$ because $v$ is
the supremum over all stopping times in \eqref{eq:value}. Hence $v=\widetilde V$.
\end{proof}

The hypothesis $b_1\le a_2$ is a condition on the separation of the two
levels, and it is explicit. The system \eqref{eq:Asystem} is homogeneous in
the skew level: $G(m)=L\,G_1(m/L)$ and $D(m)=L\,D_1(m/L)$, where $G_1,D_1$
are $G,D$ with $L=1$. Hence the band at a skew level $z$ with parameter
$\beta$ is $z\,(\hat a(\beta),\hat b(\beta))$, with
$(\hat a(\beta),\hat b(\beta))$ the band at level $1$. In particular
$b_1=L\,\hat b(\beta_L)$ and $a_2=H\,\hat a(\beta_H)$, so $b_1\le a_2$ if and
only if $H\ge H_{\mathrm c}$, where
\begin{equation}
H_{\mathrm c}:=L\,\frac{\hat b(\beta_L)}{\hat a(\beta_H)}
\label{eq:Hcrit}
\end{equation}
is the separation at which the two bands just touch.

\begin{remark}[where the additivity fails]
\label{rem:additivity}
The only case excluded above is an overlap. Write
$\widetilde V=g+(V_L-g)+(V_H-g)$, the pay-off plus two premia. On its own band each
premium satisfies $(\Lop-r)(V_i-g)=0-(\mu-r)x=(r-\mu)x>0$: it exactly offsets
the pay-off's discounted decay $(\Lop-r)g=(\mu-r)x$, which is why $\widetilde V=V_i$ is
$r$-harmonic there. Off its band a premium vanishes. When the bands are
disjoint, at most one premium is active at each point and the offset is made
once. On an overlap both are active, so the decay of the single pay-off is
offset twice:
\[
(\Lop-r)\widetilde V=(\mu-r)x+2(r-\mu)x=(r-\mu)x>0,
\]
and $\widetilde V$ ceases to be excessive. So $H_{\mathrm c}$ of \eqref{eq:Hcrit}
is also the separation below which the additive majorant
\eqref{eq:additivecap} fails.
\end{remark}

\begin{corollary}[No corner]
\label{cor:nocorner}
Under \eqref{eq:samesign} regime~C cannot occur at $H$.
\end{corollary}

\begin{proof}
Suppose it did, so that the upper boundary is $b=H$ and $v=g$ on
$[H,\infty)$, whence $v'(H^+)=1$. Being a value function, $v\ge g$
everywhere, and $v(H)=H$; hence $v-g$ is non-negative on $(a,H)$ and vanishes
at $H$, so $v'(H^-)\le1$. The local-time defect of
Theorem~\ref{thm:verification} at $H$ is then
\[
\Delta_H=(1+\beta_H)\,v'(H^+)-(1-\beta_H)\,v'(H^-)
 \ \ge\ (1+\beta_H)-(1-\beta_H)=2\beta_H>0,
\]
so $e^{-rt}v(S_t)$ fails to be a supermartingale, a contradiction.
\end{proof}

\begin{remark}
\label{rem:notglobal}
The inequality $v'(H^-)\le1$ used above is a property of the corner
configuration, not a general one. It is forced by the vanishing of $v-g$ at
$H$ from inside a continuation region. When $H$ is interior to the
continuation region, as in the merged configuration, $v-g$ does not vanish
there, and transmission gives $v'(H^-)=\pi(\beta_H)\,v'(H^+)$ with
$\pi(\beta_H)>1$, so no such bound is available.
\end{remark}

\subsection{Selection}

Two candidate configurations are available under \eqref{eq:samesign}: the
\emph{decoupled} one of Proposition~\ref{prop:decouple}, with two disjoint
bands $(a_1,b_1)$ and $(a_2,b_2)$ obtained from the two single-skew systems,
and the \emph{merged} one, a single band $(a,b)\ni L,H$ obtained from the
two-skew system \eqref{eq:master} --- the regime-B system of
Section~\ref{sec:fb}, now with $\beta_H>0$ --- with candidate
$V_{\mathrm{merged}}:=U_a$ on $(a,b)$ and $V_{\mathrm{merged}}:=g$ elsewhere.
When $b_1\le a_2$ the decoupled
configuration exists, but \eqref{eq:master} may still have a solution as
well: both systems of equations can be solvable at the same parameters. At
most one of the two candidates can be the value function, and the next result
says which.

\begin{proposition}[Selection]
\label{prop:selection}
Suppose $b_1\le a_2$ and that \eqref{eq:master} also admits a solution
$(a,b)$ with $a<L<H<b$. Then that solution is not the value function, and the
configuration of Proposition~\ref{prop:decouple} is the correct one.
\end{proposition}

\begin{proof}
Proposition~\ref{prop:decouple} identifies $v$ with $\widetilde V$, whose continuation
region is $(a_1,b_1)\cup(a_2,b_2)$ and is therefore disconnected. The merged
solution has connected continuation region $(a,b)$, so the two functions
differ; since the value function is unique, the merged one is not it.
\end{proof}

No appeal to dominance is needed: the two candidates simply have different
continuation regions. But the example behind
Proposition~\ref{prop:selection} deserves emphasis, because it shows that
\eqref{eq:master} is \emph{necessary and not sufficient}, and that a solver
gives no warning of the fact. At $\beta_L=\beta_H=0.60$, $L=1$, $H=2.2$ and
the parameters of Appendix~\ref{app:num}, the disjoint configuration is
$(0.7313,1.3113)\cup(1.6089,2.8848)$, while \eqref{eq:master} also returns the
root
\[
(a,b)=(0.94180949,\,2.84376715).
\]
This is a genuine root and not a numerical artefact: the residual of
\eqref{eq:master} is $2\times10^{-12}$, and value matching and smooth fit hold
at both endpoints to machine precision, $V(a)=a$ and $V(b)=b$ with
$V'(a)=V'(b)=1.00000000$, the transmission conditions being satisfied at $L$
and at $H$ by construction. Every condition the system encodes is met.
Nevertheless the resulting $V$ falls below the pay-off across the middle of
the band: at the midpoint $x=1.4601$ of the gap $(b_1,a_2)$ it takes the value
$1.3244$, and its worst shortfall is $V-g=-0.1360$ at $x=1.4346$.

To see what goes wrong, compare with the proof of
Proposition~\ref{prop:dominance}, writing $W=V-g$ as there. There $W$
changes monotonicity only
once, at $L$: it increases on $(a,L)$ and decreases on $(L,b)$, and the kink at
a resistance level ($\pi(\beta_H)<1$) preserves the sign of $W'$. A second
support level reverses this. With $\pi(\beta_H)>1$ the transmission condition
$V'(H^-)=\pi(\beta_H)V'(H^+)$ can turn $W'$ from negative to positive as $x$
crosses $H$ to the left, and in the example it does: $W'(H^-)=0.744$ and
$W'(H^+)=-0.563$, just as $W'(L^-)=0.118$ and $W'(L^+)=-0.720$. So $W$ has a
peak at each skew point and therefore a local minimum between them, here at
$x=1.4346$. The inequality \eqref{eq:Wode} rules out a non-negative interior
maximum but says nothing about a minimum, and the only source of value above
the pay-off is local time at a skew point, of which there is none between $L$
and $H$. Whether the minimum stays above zero therefore depends on how far
apart the two levels are, and smooth fit at $a$ and $b$ cannot decide it. This
is why dominance has to be established separately rather than read off the
system: Theorem~\ref{thm:domsame} and Corollary~\ref{cor:iff} show that it
holds exactly when $H\le H_{\mathrm c}$.

\subsection{Merger}
\label{sec:mergersub}

\begin{figure}[t]
\centering
\includegraphics[width=\textwidth]{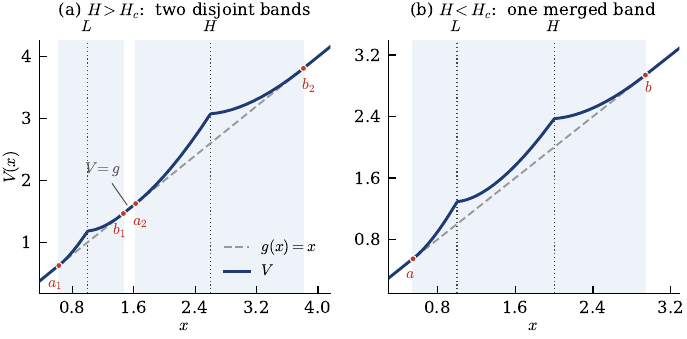}
\caption{Two support levels: $\beta_L=\beta_H=0.9$, other parameters as in
Fig.~\ref{fig:regimes}, so that $H_{\mathrm c}=2.344$ by
\eqref{eq:Hcrit}. (a) $H=2.60>H_{\mathrm c}$: the two single-skew bands
$(0.625,1.465)$ and $(1.626,3.810)$ are disjoint, and $V=g$ on the gap
between them, so the holder sells there and the levels act
independently. (b) $H=2.00<H_{\mathrm c}$: the bands have merged into the
single band $(0.547,2.939)$, which carries both skew points, and selling
between the levels is no longer optimal.}
\label{fig:merger}
\end{figure}

Fig.~\ref{fig:merger} shows the two configurations either side of
$H_{\mathrm c}$.

\begin{proposition}[Merger]
\label{prop:merger}
Fix $L,\beta_L,\beta_H$. At $H=H_{\mathrm c}$ of \eqref{eq:Hcrit}, where
$b_1=a_2$, the pair $(a_1,b_2)$ --- the union of the two touching bands ---
solves the merged system \eqref{eq:master}, and the corresponding candidate
coincides with $V_L$ on $(a_1,b_1]$ and with $V_H$ on $[a_2,b_2)$.
\end{proposition}

\begin{proof}
The $L$-band solves the single-skew system at $L$, which is
\eqref{eq:master} with the factor at $H$ replaced by the identity,
\[
T_{-\beta_L}(L)\vk(a_1)=\vk(b_1).
\] Likewise the $H$-band solves
\[
T_{-\beta_H}(H)\vk(a_2)=\vk(b_2),\qquad\text{that is}\qquad
\vk(a_2)=T_{\beta_H}(H)\vk(b_2)
\]
by \eqref{eq:Tprops}. At $H=H_{\mathrm c}$ we have $b_1=a_2$, so
\[
T_{-\beta_L}(L)\vk(a_1)=\vk(b_1)=\vk(a_2)=T_{\beta_H}(H)\vk(b_2),
\]
which is \eqref{eq:master} for $(a,b)=(a_1,b_2)$, with $a_1<L<b_1=a_2<H<b_2$.
The common vector is the coefficient pair of the candidate on $(L,H)$, and it
equals the coefficient pair of $V_L$ above $L$ and of $V_H$ below $H$, which
gives the last assertion.
\end{proof}

For $H<H_{\mathrm c}$, Proposition~\ref{lem:acomp} below shows that a
merged band exists, and Theorem~\ref{thm:domsame} that it is the continuation
region. The scalar reduction and Jacobian analysis of regime~B carry over with
$\beta_H>0$ in place of $\beta_H<0$, as explained next, and the merged band
depends continuously on $H$ (Proposition~\ref{lem:acomp}).

Proposition~\ref{prop:merger} then says that the two configurations meet at
$H_{\mathrm c}$. As $H$ decreases to $H_{\mathrm c}$ the middle component
$[b_1,a_2]$ of the stopping region shrinks to the single point $b_1=a_2$, and
at $H_{\mathrm c}$ the merged band is exactly $(a_1,b_2)$; no boundary
jumps.
Numerically, at $\beta_L=\beta_H=0.60$ and $L=1$ one finds
$H_{\mathrm c}=1.79301900$, where the bands are
$(0.731329,1.311287)$ and $(1.311287,2.351162)$ and the merged system returns
$(0.731329,2.351162)$, agreeing on both endpoints to six decimal places.

Nothing in Section~\ref{sec:fb} used the sign of $\beta_H$ except
Theorem~\ref{thm:corner} and the regime-A defect of
Theorem~\ref{thm:verification}. In particular Lemma~\ref{lem:deriv},
Corollary~\ref{lem:sens} and Proposition~\ref{prop:transition} hold verbatim
under \eqref{eq:samesign}: Corollary~\ref{lem:sens} rests on
Lemma~\ref{lem:sensneg}, whose proof uses only $\pi(\beta_L)>0$ and the
one-zero property, and that of Proposition~\ref{prop:transition} signs
$dG_H/db$ and $dD_H/db$ without reference to $\beta_H$.

Proposition~\ref{prop:dominance} is the exception, and the exception is of a
different kind from Remark~\ref{rem:signload}. There the sign convention is
load-bearing in the strict sense: with $\beta_H>0$ the regime-A conclusion is
false, because $\Delta_H=2\beta_H>0$ destroys the supermartingale property
outright. Here only the \emph{argument} fails. Its propagation step cannot
cross $H$, for the reason given before Section~\ref{sec:mergersub}: with
$\pi(\beta_H)>1$ the sign of $W'$ can flip at $H$, and $W$ then has a local
minimum between $L$ and $H$ that smooth fit does not control. The conclusion
nevertheless survives for $H<H_{\mathrm c}$ (Theorem~\ref{thm:domsame}).
Dominance is also not a side condition that could be bypassed. On one hand,
$V_{\mathrm{merged}}$ is the expected reward of the exit time of $(a,b)$, so
$V_{\mathrm{merged}}\le v$. On the other hand, if $V_{\mathrm{merged}}\ge g$,
the verification of Theorem~\ref{thm:verification} applies (with
$\Delta_H=0$, since $H\in(a,b)$) and gives $V_{\mathrm{merged}}\ge v$. So
$V_{\mathrm{merged}}\ge g$ holds if and only if $V_{\mathrm{merged}}=v$:
dominance is the whole content of optimality here, and it has to be proved
directly. Section~\ref{sec:domsame} does so by comparing with the
single-skew problem at $L$, an argument that never crosses $H$.

\begin{remark}[the additive majorant caps but does not identify]
\label{rem:cap}
Before giving the argument of Section~\ref{sec:domsame} we record how far the additive majorant
\eqref{eq:additivecap} goes, since it yields a computable bound of independent
use. It rests on the characterisation of \citet{DayanikKaratzas03}. Let
$F:=\psi_r/\varphi_r$, which is strictly increasing by
Lemma~\ref{lem:onezero}, and for a function $u$ on $(0,\infty)$ write
$\widehat u:=(u/\varphi_r)\circ F^{-1}$, a function of the new coordinate
$y=F(x)$. For a function $f$ of $y$, let $\mathcal C f$ denote its
\emph{least concave majorant}, the smallest concave function lying above
$f$. Then a non-negative $u$ is $r$-excessive if and only if $\widehat u$ is
concave, and the value function is $v=\varphi_r\cdot(\mathcal C\widehat
g)\circ F$, that is $\widehat v=\mathcal C\widehat g$. Now
$\widehat{\widetilde V}\ge\widehat g$ by Proposition~\ref{prop:decouple}'s first step,
which holds irrespective of overlap; hence $\mathcal C\widehat{\widetilde V}$ is concave
and dominates $\widehat g$, so
\begin{equation}
v\ \le\ \varphi_r\cdot\bigl(\mathcal C\widehat{\widetilde V}\bigr)\circ F
\label{eq:concavecap}
\end{equation}
unconditionally --- a cap even where $\widetilde V$ itself is not excessive. Moreover
$\mathcal C\widehat{\widetilde V}=\widehat v$ if and only if $\widetilde V\le v$, since $\widehat v$
is concave and $\mathcal C\widehat{\widetilde V}\ge\mathcal C\widehat g=\widehat v$
always. That last condition is not automatic. For moderate parameters it holds
and \eqref{eq:concavecap} is exact: at $\beta_L=\beta_H=0.60$, $H=1.60$ and at
$\beta_L=0.60,\beta_H=0.30$, $H=1.20$ we find $\mathcal C\widehat{\widetilde V}=\widehat v$ to machine precision. But $\widetilde V$ can overshoot the value function
when the skews are strong and the levels close. In both examples below
$H<H_{\mathrm c}$, since \eqref{eq:Hcrit} gives $H_{\mathrm c}=2.446$ for
$\beta_L=\beta_H=0.95$ and $H_{\mathrm c}=2.344$ for $\beta_L=\beta_H=0.90$;
so $V_{\mathrm{merged}}=v$ by Theorem~\ref{thm:domsame}, and the overshoot is
$\max_x\bigl(\widetilde V-V_{\mathrm{merged}}\bigr)$, attained near $x=H$: at
$\beta_L=\beta_H=0.95$, $H=1.20$ it is $0.047$ (for instance
$\widetilde V(1.2002)=1.5059$ against $V_{\mathrm{merged}}(1.2002)=1.4586$),
and at $\beta_L=\beta_H=0.90$, $H=1.25$ it is $6.4\times10^{-3}$. In these cases $\widetilde V>v$ somewhere, so by the equivalence
above the bound \eqref{eq:concavecap} lies strictly above $v$. In general,
knowing whether the bound is exact requires checking $\widetilde V\le v$,
which already presupposes knowledge of $v$. So the additive majorant gives a
computable upper bound on $v$, but it cannot by itself identify $v$.
\end{remark}

\subsection{Dominance on the merged band}
\label{sec:domsame}

The obstruction described in Section~\ref{sec:mergersub} is that the
propagation step of Proposition~\ref{prop:dominance} transports information
\emph{across} $H$. The argument below never does. It
compares the merged candidate with the single-skew candidate at $L$, and the
comparison is carried entirely by Lemma~\ref{lem:onezero}: a solution
vanishing once cannot vanish again. Recall from
Section~\ref{sec:prelim} that $U_a$ is the solution with smooth fit at $a$,
that $V=U_{a}$ on the band, and that $w_a=\partial U_a/\partial a<0$ on
$(a,\infty)$ by Lemma~\ref{lem:sensneg}. Write $U^0_a$ for the corresponding
solution built with the skew at $L$ alone, so that $U_a=U^0_a$ on $(a,H)$ and
the two differ above $H$ only.

\begin{lemma}[a positive skew damps the candidate]
\label{lem:damping}
Under \eqref{eq:samesign}, $U_a<U^0_a$ on $(H,\infty)$.
\end{lemma}

\begin{proof}
Define the difference $d:=U_a-U^0_a$ on $(a,\infty)$. The two functions
agree on $(a,H]$, so $d=0$ there and in particular $d(H)=0$. At $H$ the
transmission condition gives $U_a'(H^+)=U_a'(H^-)/\pi(\beta_H)$, whereas
$U^0_a$ has no kink there, so $U^{0\prime}_a(H)=U_a'(H^-)$. Here $U_a>0$ on
$(a,\infty)$, since a zero would force strict monotonicity by
Lemma~\ref{lem:onezero} while $U_a'(a)=1>0$ makes $U_a$ increase from
$U_a(a)=a>0$; and $U_a'>0$ on $(a,H)$ by the argument of
Lemma~\ref{lem:incr}, a critical point of a positive solution satisfying
$\tfrac12\sigma^2x^2U_a''=rU_a>0$ and so being a strict minimum, while the
transmission condition at $L$ rescales the derivative by the positive factor
$\pi(\beta_L)^{-1}$. Hence, with $\pi(\beta_H)>1$ and $U_a'(H^-)>0$,
\[
d'(H^+)=U_a'(H^-)\bigl(\pi(\beta_H)^{-1}-1\bigr)<0.
\]
On $(H,\infty)$ both functions solve $(\Lop-r)u=0$ with no skew point, so $d$
does too; by Lemma~\ref{lem:onezero} its only zero there is at $H$ and it is
strictly monotone, hence $d<0$ on $(H,\infty)$.
\end{proof}

\begin{proposition}[existence of the merged band]
\label{lem:acomp}
Assume \eqref{eq:samesign}, let $H\in(L,H_{\mathrm c})$, and let $a_1$ be
the lower endpoint of the $L$-band. Then \eqref{eq:master} has a solution
$(a,b)=(a(H),b(H))$ with
\[
a(H)<a_1<L<H<b(H)
\qquad\text{and}\qquad
U_{a(H)}\ge g\ \text{ on }[H,\infty).
\]
It is the only solution of \eqref{eq:master} with $U_a\ge g$ on
$[H,\infty)$. At $H=H_{\mathrm c}$ the same holds with
$(a(H_{\mathrm c}),b(H_{\mathrm c}))=(a_1,b_2)$, by
Proposition~\ref{prop:merger}, and $H\mapsto(a(H),b(H))$ is continuous on
$(L,H_{\mathrm c}]$.
\end{proposition}

\begin{proof}
See Appendix~\ref{app:proofs}.
\end{proof}

\begin{theorem}[dominance, same-sign case]
\label{thm:domsame}
Assume \eqref{eq:samesign} and $H<H_{\mathrm c}$, and let
$(a,b)=(a(H),b(H))$ be the solution of \eqref{eq:master} given by
Proposition~\ref{lem:acomp}. Then $V>g$ on $(a,b)$, and consequently $V=v$ and the exit
time of $(a,b)$ is optimal.
\end{theorem}

\begin{proof}
On $(a,L)$ and on $(H,b)$ the two propagation steps in the proof of
Proposition~\ref{prop:dominance} apply verbatim: they use only \eqref{eq:Wode},
smooth fit at the adjacent free boundary, and Assumption~\ref{ass:drift}, and
nowhere the sign of $\beta_H$. This gives $W>0$ on $(a,L]$ and on $[H,b)$.

It remains to treat $[L,H]$. There $V=U_{a(H)}$, and this depends on
$(a(H),L,\beta_L)$ alone, not on $H$ or $\beta_H$. Since
$H<H_{\mathrm c}$ we have $[L,H]\subset[L,H_{\mathrm c}]$, and
Proposition~\ref{lem:acomp} gives $a(H)<a_1$, so Lemma~\ref{lem:sensneg} yields
\[
U_{a(H)}\ \ge\ U_{a_1}\qquad\text{on }[L,H].
\]
At $H=H_{\mathrm c}$, Proposition~\ref{prop:merger} identifies the merged
solution with the union of the two touching bands, so on $(L,b_1)$ the
function $U_{a_1}$ is the single-skew candidate at $L$, which exceeds $g$ by
Proposition~\ref{prop:dominance}. On $(a_2,H_{\mathrm c})$ it coincides
with the single-skew candidate at $H$: both solve the equation there, and at
the touch point $b_1=a_2$ both take the value $b_1$ with derivative $1$, so
they are equal. Hence $U_{a_1}\ge g$ on $(L,H_{\mathrm c})$, with equality
only at $b_1=a_2$, and therefore $U_{a(H)}>g$ on $[L,H]$ whenever
$a(H)<a_1$, that is whenever $H<H_{\mathrm c}$.

Given $V>g$, the verification of Theorem~\ref{thm:verification} applies without
change: in the merged configuration $H\in(a,b)$, so $\Delta_H=0$ by
\eqref{eq:transmission} and the sign of $\beta_H$ is not used anywhere.
\end{proof}

\begin{corollary}[the criterion is sharp]
\label{cor:iff}
Under \eqref{eq:samesign}, the merged candidate dominates the pay-off if and
only if $H\le H_{\mathrm c}$.
\end{corollary}

\begin{proof}
Sufficiency is Theorem~\ref{thm:domsame}. For necessity, if
$H>H_{\mathrm c}$ then $b_1<a_2$ and Proposition~\ref{prop:decouple}
produces a candidate with a disconnected continuation region, which by
Proposition~\ref{prop:selection} is the value function; the merged candidate is
then a different function with $V_{\mathrm{merged}}\le v$ and
$V_{\mathrm{merged}}\ne v$, so $V_{\mathrm{merged}}\not\ge g$.
\end{proof}

\subsection{Identification inverts}

\begin{corollary}
\label{cor:identsame}
Under \eqref{eq:samesign} with $b_1\le a_2$, both $\beta_L$ and $\beta_H$ are
identified from the optimal rule: the first from $(a_1,b_1)$ and the second
from $(a_2,b_2)$, by Theorem~\ref{thm:inversion} applied to each band
separately.
\end{corollary}

Against Theorem~\ref{thm:ident} this gives a dichotomy. When the two skews
have opposite signs, the resistance parameter is unidentified on most of its
range --- invisible in regime~A and constant in its effect throughout regime~C
--- and only the location of resistance can be recovered. When they have the
same sign and the bands are disjoint, both parameters are separately and
exactly identified, because the two problems do not interact at all.

\subsection{What the sign convention carries}
\label{sec:signcarries}

Section~\ref{sec:samesign} is mathematically the richer case --- a
three-component stopping region, a merger bifurcation, two separately
identified parameters --- and economically the flatter one. The reason is
Corollary~\ref{cor:nocorner}. With $\beta_H>0$ there is no corner, so the
optimal boundaries $a$ and $b$ sit at no distinguished price: they are
interior points fixed by smooth fit, and a reader of the chart would see
thresholds bearing no relation to either drawn line. The rule a practitioner
would recognise --- sell \emph{exactly} at $H$, the content of
Theorem~\ref{thm:corner} --- requires both
\[
\beta_H<0
\qquad\text{and}\qquad
\pi(\beta_H)\le V'(H^-),
\]
and the first of these is exactly what \eqref{eq:samesign} gives up.

This locates the technical-analysis content of the paper. It is not the
local-time mechanism that produces a level-based rule: local time at $L$
produces a band around $L$, whichever way it pushes. What produces the rule is
the \emph{opposition} of the two signs, which makes one level a place the
holder waits at and the other a place the holder will not cross. Support and
resistance are not two instances of one phenomenon here; they are the two
signs of one parameter, and only their combination yields a rule of the
practitioner's form.

\section{The other sign pairs}
\label{sec:othersigns}

The model has four sign pairs, of which Sections~\ref{sec:fb}--\ref{sec:ident}
treat $\beta_L>0>\beta_H$ and Section~\ref{sec:samesign} treats
$\beta_L,\beta_H>0$. This section disposes of the remaining two.

\begin{proposition}[no band]
\label{prop:bothneg}
If $\beta_L\le0$ and $\beta_H\le0$ then $v=g$ and immediate liquidation is
optimal.
\end{proposition}

\begin{proof}
With $g(x)=x$ the local-time defect at either level is
$\Delta_z(g)=(1+\beta_z)-(1-\beta_z)=2\beta_z\le0$, and
$(\Lop-r)g=(\mu-r)x<0$ by Assumption~\ref{ass:drift}. So $g$ is
$r$-superharmonic on $(0,\infty)$ and, being non-negative, $r$-excessive.
A non-negative $r$-excessive majorant of $g$ dominates $v$, and $g$ majorises
itself, so $v\le g$; the reverse holds because $\tau=0$ is admissible.
\end{proof}

This is the two-level form of Corollary~\ref{cor:noband}: a continuation
region requires a level that pushes the price in the direction in which the
pay-off increases, and neither level does.

\subsection{Support above, breakdown below}

The remaining pair is $\beta_L<0<\beta_H$: the lower level pushes the price
down and the upper one pushes it up. Now $H$ is the level that creates value
and $L$ the level the holder would rather avoid, so the geometry of
Sections~\ref{sec:fb}--\ref{sec:verify} recurs with the roles of the two
levels exchanged. Write $(a_2,b_2)$ for the single-skew band at $H$, that is
the band of Theorem~\ref{thm:Airrelevant} with $(L,\beta_L)$ replaced by
$(H,\beta_H)$.

\begin{proposition}[support is invisible]
\label{prop:Lirrelevant}
Let $\beta_L<0<\beta_H$ and suppose $a_2\ge L$. Then the continuation region
is $(a_2,b_2)$, the value function is the corresponding candidate, and both
are independent of $\beta_L$.
\end{proposition}

\begin{proof}
Let $V$ be that candidate: the single-skew solution at $H$ on $(a_2,b_2)$ and
$g$ elsewhere. Since $a_2\ge L$, the band contains no skew point other than
$H$, so $V$ solves $(\Lop-r)V=0$ there in the sense of
Lemma~\ref{lem:transmission} with the transmission condition at $H$, and
$V>g$ on $(a_2,b_2)$ by Proposition~\ref{prop:dominance} applied to the
single-skew problem at $H$. For the verification of
Theorem~\ref{thm:verification}: $\Delta_H=0$ because $H\in(a_2,b_2)$;
$\Delta_L=2\beta_L<0$ because $V=g$ on a neighbourhood of $L$ when $a_2>L$,
and $\Delta_L=(1+\beta_L)V'(L^+)-(1-\beta_L)=2\beta_L<0$ when $a_2=L$, where
smooth fit gives $V'(L^+)=1$; and $(\Lop-r)V\le0$ off the band as before.
Hence $V=v$.
\end{proof}

\begin{remark}[the remaining two geometries]
\label{rem:mirror}
When $a_2<L$ the band straddles $L$ and $\beta_L$ enters, through
\eqref{eq:master} exactly as $\beta_H$ does in regime~B; and the lower
boundary may pin at $a=L$, the mirror of regime~C, in which case the holder
sells \emph{exactly at $L$} --- a stop-loss at a level that pushes the price
down. The criterion is the mirror of \eqref{eq:criterion}: with $V=g$ below
$L$ the defect there is $\Delta_L=(1+\beta_L)V'(L^+)-(1-\beta_L)$, so the
corner obtains exactly when
\[
\beta_L\ \le\ \beta_L^\star:=\frac{1-V'(L^+)}{1+V'(L^+)},
\]
to be compared with $\beta_H^\star=(V'(H^-)-1)/(V'(H^-)+1)$ of
Theorem~\ref{thm:corner}. The other results of
Sections~\ref{sec:fb}--\ref{sec:ident} have mirrors in the same sense:
Theorem~\ref{thm:corner} becomes the criterion just displayed,
Proposition~\ref{prop:transition} the continuity of the transition at
$\beta_L^\star$, and Theorem~\ref{thm:ident} the statement that $\beta_L$ is
recoverable from the optimal rule only for $\beta_L\in(\beta_L^\star,0)$,
being invisible when the band lies above $L$ and constant in its effect once
the boundary pins there. At $L=1$, $H=1.2$ one finds
$\beta_L^\star=-0.312$ for $\beta_H=0.90$ and $\beta_L^\star=-0.137$ for
$\beta_H=0.60$. The analysis of these two geometries follows
Sections~\ref{sec:fb}--\ref{sec:verify} with the two levels exchanged, and we
do not repeat it.
\end{remark}

The four pairs therefore divide as follows. Opposite signs with support below
and resistance above, $\beta_L>0>\beta_H$, push the price into $(L,H)$ from
both ends and produce the practitioner's rule of Theorem~\ref{thm:corner}.
Opposite signs the other way round, $\beta_L<0<\beta_H$, push it out of
$(L,H)$ and produce a stop-loss rule at $L$. Two positive signs give two
bands, one around each level, merging at $H_{\mathrm c}$. Two negative signs
give no band at all.

\section{Conclusion}
\label{sec:disc}

We have solved the perpetual liquidation problem for a geometric multi-skew
Brownian motion with a support level $L$ and a resistance level $H>L$. The
continuation region is an interval $(a,b)$ containing $L$, and its position
relative to $H$ falls into three regimes separated by the closed-form
criterion \eqref{eq:criterion}. When $b\le H$ the resistance parameter is
exactly irrelevant, and the solution is that of the single-skew problem at
$L$ (Theorem~\ref{thm:Airrelevant}); when the criterion fails the band
straddles $H$, with the upper boundary pinned at $H$ over a closed interval
of resistance strengths before smooth fit resumes above it
(Theorem~\ref{thm:corner}, Proposition~\ref{prop:transition}). The candidate
so constructed is the value function
(Theorem~\ref{thm:verification}). Selling into resistance is therefore a
conclusion rather than an assumption.

The identification consequences are asymmetric
(Theorem~\ref{thm:ident}). Observed exercise determines the support
parameters exactly and in closed form, whereas resistance is recoverable only
up to an interval whose endpoint $\beta_H^\ast$ is explicit, and not at all
when the band lies below $H$: a level the holder arranges not to visit leaves
no trace in the stopping rule. When the two skews share a sign the problem
decouples (Proposition~\ref{prop:decouple}), one band dominating the other
below a critical separation $H_c$ and the two acting independently above it
(Proposition~\ref{prop:merger}, Theorem~\ref{thm:domsame}), and both
parameters are then identified (Corollary~\ref{cor:identsame}).

\paragraph{Where the value comes from.}
Under Assumption~\ref{ass:drift} the discounted price is a supermartingale off
the skew points, so absent local time the holder sells at once
(Corollary~\ref{cor:noband}). Every continuation region in the paper is
generated by the support push alone, and the short-horizon comparison
$\E[\ell_t]\sim\sqrt t$ against $O(t)$ is the reason: a pay-off cannot be
$r$-excessive at a skew point where $\beta g'>0$, whatever the drift. This is
the mechanism of \citet{AlvarezSalminen17}, and it explains uniformly what
would otherwise be a case analysis.

\paragraph{Disconnected stopping.}
The stopping region is disconnected. For a linear pay-off under ordinary
geometric Brownian motion this cannot happen, and for skew Brownian motion it
can look like an artefact of particular parameters. Here it is unavoidable,
and for a structural reason: the only source of value above the pay-off is
local time at a support level, which acts at a point, so a continuation
region is always a bounded interval around such a level. The stopping region
is $(0,a]\cup[b,\infty)$ throughout Sections~\ref{sec:fb}--\ref{sec:ident}
and $(0,a_1]\cup[b_1,a_2]\cup[b_2,\infty)$ when the two skews share a sign
and the bands stay apart.

\paragraph{Estimation.}
Theorem~\ref{thm:ident} constrains what can be learned from exercise decisions,
not from the price path. The threshold-estimator literature for skewed and
oscillating diffusions recovers skew parameters from high-frequency occupation
times, and nothing here obstructs that; the point is that the two data sources
are not interchangeable.

\paragraph{Role reversal.}
The model has no mechanism by which a level changes its role: the character of
each level is fixed by the sign of its skew parameter, and the state is the
price alone. Role reversal at a single level, as in \citet{JackaMaeda20} and
\citet{Henderson26}, requires an auxiliary state recording where the price has
been. The two classes of model are complementary --- flag models describe one
level whose role depends on the path, the present model two levels whose roles
do not --- and a natural extension combines them, with a flag switching the
skew parameters at each of two levels.

\paragraph{Extensions.}
Besides the combination just described, finite maturity seems worthwhile: the
boundaries become curves and the criterion \eqref{eq:criterion} presumably
becomes a time-dependent inequality.

\newcommand{\aideclaration}{In preparing this work the author used
Anthropic's Claude as an assistant for drafting and revising the exposition,
for checking and helping to complete several proofs, for the numerical
verification reported in Appendix~\ref{app:num}, for the figures, and for
\LaTeX\ preparation. All definitions, statements and proofs were checked by
the author, who reviewed and edited the output and takes full responsibility
for the content of the paper.}

\ifarxiv
  \section*{Declaration on the use of generative AI}

  \aideclaration
\fi

\appendix

\section{Two proofs}
\label{app:proofs}

\begin{proof}[Proof of Proposition~\ref{prop:transition}]
The claim is about the single scalar $\pi=B(a(b))/D_H(b)$ of
\eqref{eq:scalar} viewed as a function of $b$ on $(H,\bar b)$. We show it is
strictly increasing, then evaluate it at each end.

\smallskip
\emph{Step 1: derivatives of $G_H$ and $D_H$.}
Substituting $t:=H/b\in(0,1)$ gives $G_H(b)=b\,P(t)$ and $D_H(b)=b\,Q(t)$
with $P(t)=\kappa_1t^{\alpha_1}+\kappa_2t^{\alpha_2}$ and
$Q(t)=\kappa_1\alpha_1t^{\alpha_1}+\kappa_2\alpha_2t^{\alpha_2}$. By
\eqref{eq:vkprime},
\[
\frac{dG_H}{db}=c\bigl(t^{\alpha_1}-t^{\alpha_2}\bigr)>0,
\qquad
\frac{dD_H}{db}=c\bigl(\alpha_1t^{\alpha_1}-\alpha_2t^{\alpha_2}\bigr)<0,
\]
the first because $t<1$ and $\alpha_1<0<\alpha_2$ give
$t^{\alpha_1}>1>t^{\alpha_2}$, the second because both terms are negative.

\smallskip
\emph{Step 2: $\pi$ is strictly increasing in $b$.}
The obstacle to differentiating \eqref{eq:scalar} directly is that its
numerator $B$ is a function of $a$ and its denominator $D_H$ a function of
$b$. Re-parameterise by a variable in which both are functions of the same
argument: continuity at $H$ says $V(H^-)=V(H^+)$, and we call this common
value
\[
y:=A(a)=G_H(b),
\]
the level the candidate attains at the resistance point. Since $dG_H/db>0$ the
map $b\mapsto y$ is a strictly increasing bijection, and by
Corollary~\ref{lem:sens} $a\mapsto y$ is a strictly decreasing one. Both $B$
and $D_H$ are therefore functions of $y$ alone, with
\[
\frac{dB}{dy}=\frac{dB/da}{dA/da}>0,
\qquad
\frac{dD_H}{dy}=\frac{dD_H/db}{dG_H/db}<0 ,
\]
the first by Corollary~\ref{lem:sens} (both factors negative), the second by
Step~1. Moreover $B,D_H>0$, being $H$ times the one-sided derivatives of $V$
at $H$, which are positive by Lemma~\ref{lem:incr}. Applying the quotient
rule to $\pi=B/D_H$ in \eqref{eq:scalar}, now with both $B$ and $D_H$ regarded
as functions of $y$, gives
\[
\frac{d\pi}{dy}=\frac{(dB/dy)\,D_H-B\,(dD_H/dy)}{D_H^2}>0,
\]
since both $(dB/dy)\,D_H$ and $-B\,(dD_H/dy)$ are positive. Finally, the
chain rule gives $d\pi/db=(d\pi/dy)(dy/db)$ with $dy/db=dG_H/db>0$, so
$d\pi/db>0$.

\smallskip
\emph{Step 3: the value of $\pi$ at $b=H$.}
Here $t=1$, so $P(1)=Q(1)=1$ by \eqref{eq:kappa} and
$\kappa_1\alpha_1+\kappa_2\alpha_2=1$, giving $G_H(H)=D_H(H)=H$. Hence
$\pi=B(a)/H=V'(H^-)$, the equality case of \eqref{eq:criterion}.

\smallskip
\emph{Step 4: the value of $\pi$ at $b=\bar b$.}
We claim $\pi=1$ there, for the simple reason that the candidate has no kink
at $H$ at all. (Recall that $\beta_H$ is not given along the branch but
determined by $b$ through \eqref{eq:scalar}; the claim is that the value it
takes at this endpoint is zero.) Take the $L$-band $(\bar a,\bar b)$, for which the skew at $H$
is absent, and compare the coefficient pairs on the two sides of $H$. Below,
on $(L,H)$, the pair is $\Phi(\bar a)$: reaching it from the lower boundary
requires crossing $L$, which is what the matrix in \eqref{eq:Phi} does. Above,
on $(H,\bar b)$, the pair is simply $\vk(\bar b)$, with no matrix at all,
because $(H,\bar b)$ contains no skew point and smooth fit at $\bar b$ already
fixes the coefficients there. Since $T_0=I$ by \eqref{eq:Tprops}, crossing $H$
acts as the identity, so these two pairs are one and the same vector:
\[
\Phi(\bar a)=\vk(\bar b).
\]
Now contract this one identity with each of the two functionals of
\eqref{eq:ABGD} in turn. Against $e_H$ it gives
\[
A(\bar a)=G_H(\bar b),
\]
which is value matching, $V(H^-)=V(H^+)$, at $b=\bar b$. Since $a(b)$ was
defined as $A^{-1}(G_H(b))$ and $A$ is injective, this says exactly that
$a(\bar b)=\bar a$. In terms of the regime-B branch
$b\mapsto(a(b),b,\beta_H(b))$ defined after \eqref{eq:scalar}: that its
$b$-coordinate ends at $\bar b$ holds by construction, but that its
$a$-coordinate ends at $\bar a$ is what this identity adds, so the branch
terminates at the $L$-band $(\bar a,\bar b)$ itself. Against $e_H'$ the same identity gives
\[
B(\bar a)=D_H(\bar b),
\]
that is $HV'(H^-)=HV'(H^+)$, so the two one-sided slopes at $H$ agree.
Therefore, evaluating \eqref{eq:scalar} at $b=\bar b$, substituting
$a(\bar b)=\bar a$, and using that the numerator and denominator are then the
same non-zero number,
\[
\pi=\frac{B\bigl(a(\bar b)\bigr)}{D_H(\bar b)}
 =\frac{B(\bar a)}{D_H(\bar b)}=1,
\]
and $\pi=1$ reads $(1+\beta_H)/(1-\beta_H)=1$, so $\beta_H=0$. The last
conclusion can also be read straight off the transmission condition
\eqref{eq:transmission}: with $V'(H^+)=V'(H^-)\ne0$ --- non-zero by
Lemma~\ref{lem:incr} --- the two sides of
$(1+\beta_H)V'(H^+)=(1-\beta_H)V'(H^-)$ cancel to $1+\beta_H=1-\beta_H$. Note
that only the slope equality enters here; the value equality is doing the
separate work of identifying $a(\bar b)$ with $\bar a$.

\smallskip
\emph{Step 5: conclusion.}
By Step~2 the map $b\mapsto\pi$ is continuous and strictly increasing on
$(H,\bar b)$, and by Steps~3 and~4 its values at the two ends are $V'(H^-)$
and $1$. Strict monotonicity therefore forces
\[
V'(H^-)<1,
\]
and the intermediate value theorem makes $b\mapsto\pi$ a bijection from
$(H,\bar b)$ onto $(V'(H^-),1)$.

To transfer this to $\beta_H$, invert the permeability: solving
$\pi=(1+\beta)/(1-\beta)$ gives $\beta=(\pi-1)/(\pi+1)$, a strictly increasing
bijection of $(0,\infty)$ onto $(-1,1)$. Composing, $b\mapsto\beta_H$ is a
strictly increasing bijection of $(H,\bar b)$ onto the image of
$(V'(H^-),1)$, whose endpoints are
\[
\frac{V'(H^-)-1}{V'(H^-)+1}=\beta_H^\star
\qquad\text{and}\qquad
\frac{1-1}{1+1}=0,
\]
the first being exactly the quantity named in Theorem~\ref{thm:corner}. The
image is thus $(\beta_H^\star,0)$, and $V'(H^-)<1$ shows in passing that
$\beta_H^\star<0$, so this interval is non-empty and lies in $(-1,0)$ as the
sign convention \eqref{eq:betasigns} requires. Finally, the Jacobian. Contracting the system \eqref{eq:master} with $e_H$
and $e_H'$ is a change of basis by the matrix $M$ with those two rows, whose
determinant $H^{\alpha_1+\alpha_2}(\alpha_2-\alpha_1)$ is non-zero, and in the
new basis the system reads $A(a)-G_H(b)=0$, $B(a)-\pi D_H(b)=0$. Its Jacobian
in $(a,b)$ has determinant
$\frac{dA}{da}\bigl(-\pi\frac{dD_H}{db}\bigr)
+\frac{dG_H}{db}\frac{dB}{da}$, and substituting
$\frac{da}{db}=\frac{dG_H}{db}\big/\frac{dA}{da}$ from the first equation into
$\frac{d\pi}{db}=\bigl(\frac{dB}{da}\frac{da}{db}D_H-B\frac{dD_H}{db}\bigr)/D_H^2$,
together with $B=\pi D_H$ on the branch, turns this into
\begin{equation}
\det J=\frac{1}{H^{\alpha_1+\alpha_2}(\alpha_2-\alpha_1)}\,
\frac{d\pi}{db}\,\frac{dA}{da}\,D_H .
\label{eq:detJpi}
\end{equation}
The last three factors are non-zero --- $dA/da<0$ by
Corollary~\ref{lem:sens} and $D_H>0$ by Step~2 --- so $\det J\ne0$ is
equivalent to $d\pi/db\ne0$, which is Step~2. Incidentally \eqref{eq:detJpi}
also fixes the sign: $d\pi/db>0$, $dA/da<0$, $D_H>0$ and $\alpha_2>\alpha_1$
give $\det J<0$ throughout the branch.

\smallskip
\emph{Step 6: transversality at $b=H$.}
It remains to show that $d\beta_H/db$ has a finite, positive limit as
$b\downarrow H$. The functions $A,B$ are smooth in $a$ and $G_H,D_H$ are
smooth in $b$ on a neighbourhood of $H$, so the expression for $d\pi/db$ used
above extends continuously to $b=H$, where $a(b)\to a_H:=A^{-1}(G_H(H))$ with
$dA/da(a_H)<0$ by Corollary~\ref{lem:sens}. At $b=H$ we have $t=1$, so Step~1
gives $dG_H/db=0$ and $dD_H/db=-c(\alpha_2-\alpha_1)$. Hence
$da/db=(dG_H/db)/(dA/da)=0$ there, and therefore
$dB/db=(dB/da)(da/db)=0$: the first term in the numerator of $d\pi/db$
vanishes in the limit, since $dB/da$ is finite at $a_H$. With $D_H(H)=H$ and
$B(a_H)=H\,V'(H^-)$ from Step~3, only the second term survives:
\[
\lim_{b\downarrow H}\frac{d\pi}{db}
=-\frac{B(a_H)}{D_H(H)^2}\,\frac{dD_H}{db}(H)
=\frac{c(\alpha_2-\alpha_1)\,V'(H^-)}{H}.
\]
Since $d\beta/d\pi=2/(1+\pi)^2$ and $\pi=V'(H^-)$ at $b=H$,
\begin{equation}
\lim_{b\downarrow H}\frac{d\beta_H}{db}
=\frac{2c(\alpha_2-\alpha_1)\,V'(H^-)}{H\bigl(1+V'(H^-)\bigr)^2}
\in(0,\infty).
\label{eq:transversal}
\end{equation}
\end{proof}

\begin{proof}[Proof of Proposition~\ref{lem:acomp}]
\emph{Positivity of the crossing matrix.} For $\beta\in(0,1)$ every entry of
$T_{-\beta}(z)$ is positive: replacing $\beta$ by $-\beta$ in \eqref{eq:T},
the denominator is $(\alpha_2-\alpha_1)(1+\beta)>0$ and the entries of the
numerator are $(\alpha_2-\alpha_1)+(\alpha_1+\alpha_2)\beta$,
$2\alpha_2\beta z^{\alpha_2-\alpha_1}$, $-2\alpha_1\beta z^{\alpha_1-\alpha_2}$
and $(1-\beta)\alpha_2-(1+\beta)\alpha_1$, all positive because
$\alpha_1<0<\alpha_2$ and $|\alpha_1+\alpha_2|<\alpha_2-\alpha_1$. Hence
$K:=T_{-\beta_H}(H)T_{-\beta_L}(L)$ has positive entries, and for $0<a<L$ the
coefficient pair of $U_a$ on $(H,\infty)$,
\[
\gamma(a):=K\,\vk(a),
\]
has both components positive. Since $\alpha_2>1$, $U_a(x)-x\to\infty$ as
$x\to\infty$, so
\[
h(a):=\min_{x\ge H}\bigl(U_a(x)-x\bigr)
\]
is attained. By Lemma~\ref{lem:sensneg}, $h$ is continuous and strictly
decreasing: if the minimum for $a'$ is attained at $x'$ and $a>a'$, then
$h(a)\le U_a(x')-x'<U_{a'}(x')-x'=h(a')$.

\emph{$h>0$ for small $a$.} Write $(p,q)^{\!\top}$ for the second column of
$K$, whose entries are positive, and
$\mu_0:=\inf_{x>0}\bigl(px^{\alpha_1-1}+qx^{\alpha_2-1}\bigr)$. The
infimand is continuous and strictly positive on $(0,\infty)$, and it tends to
$\infty$ at both ends --- as $x\downarrow0$ through its first term, since
$\alpha_1-1<0$, and as $x\to\infty$ through its second, since
$\alpha_2-1>0$. It therefore attains its infimum at some interior point,
where its value is strictly positive, so $\mu_0>0$. Now
$\vk(a)$ has components $\kappa_1a^{1-\alpha_1}$ and
$\kappa_2a^{1-\alpha_2}$, so $\gamma(a)=K\vk(a)$ is the combination of the
columns of $K$ with those two weights, whence
\[
\gamma(a)\ \ge\ \kappa_2a^{1-\alpha_2}(p,q)^{\!\top}
\]
componentwise, the discarded term being a positive multiple of the first
column of $K$, which is also positive. Since $x^{\alpha_1},x^{\alpha_2}>0$, this gives
$U_a(x)\ge\kappa_2a^{1-\alpha_2}(px^{\alpha_1}+qx^{\alpha_2})
\ge\kappa_2a^{1-\alpha_2}\mu_0\,x$ on $(H,\infty)$, and this exceeds $x$ once
$\kappa_2a^{1-\alpha_2}\mu_0>1$, which holds for small $a$ because
$1-\alpha_2<0$.

\emph{$h(a_1)<0$.} If $H<b_1$, Lemma~\ref{lem:damping} gives
$U_{a_1}(b_1)<U^0_{a_1}(b_1)=b_1$ with $b_1>H$. If $b_1\le H$, work in the
model with the skew at $H$ alone and let $u_m$ be its solution whose
coefficient pair on $(0,H)$ is $\vk(m)$; the subscript is thus always a
point below $H$, at which $u_m$ has smooth fit, and since $a_2<H<b_2$ the
$H$-band is $u_{a_2}$, not $u_{b_2}$. The argument of
Lemma~\ref{lem:sensneg}, with $m$ in place of $a$, shows that
$m\mapsto u_m(x)$ is strictly decreasing for $x>m$. The function $U_{a_1}$
carries $\vk(a_1)$ on $(a_1,L)$, and crossing $L$ turns that pair into
$\vk(b_1)$ by the $L$-band equation $T_{-\beta_L}(L)\vk(a_1)=\vk(b_1)$;
since $b_1\le H$ this is its pair on $(L,H)$, and crossing $H$ applies
$T_{-\beta_H}(H)$ to it, exactly as for $u_{b_1}$. Hence $U_{a_1}=u_{b_1}$
on $(L,\infty)$, while the single-skew candidate at $H$ is $u_{a_2}$, with
$u_{a_2}(b_2)=b_2$. Now $H<H_{\mathrm c}$ means $a_2<b_1$
by \eqref{eq:Hcrit}, and $b_2>H\ge b_1$, so
$U_{a_1}(b_2)=u_{b_1}(b_2)<u_{a_2}(b_2)=b_2$. In both cases $h(a_1)<0$.

\emph{Conclusion.} By the intermediate value theorem there is a unique
$a^\ast\in(0,a_1)$ with $h(a^\ast)=0$. The minimum is not attained at $x=H$:
by Lemma~\ref{lem:sensneg} applied to $U^0$,
$U_{a^\ast}(H)=U^0_{a^\ast}(H)>U^0_{a_1}(H)\ge H$, the last inequality
because $U^0_{a_1}\ge g$ on $(a_1,\infty)$ --- on $(a_1,b_1)$ by
Proposition~\ref{prop:dominance}, and on $[b_1,\infty)$ because the
function with coefficient pair $\vk(m)$ equals $m\,\chi(x/m)$ with
$\chi(t):=\kappa_1t^{\alpha_1}+\kappa_2t^{\alpha_2}$; here $\chi(t)-t$ is
strictly convex, since $\alpha_1(\alpha_1-1)>0$ and $\alpha_2(\alpha_2-1)>0$,
and vanishes together with its derivative at $t=1$, because
$\kappa_1+\kappa_2=1$ and $\kappa_1\alpha_1+\kappa_2\alpha_2=1$ by
\eqref{eq:kappa}. Hence $\chi(t)-t\ge0$ for all $t>0$, so
$m\,\chi(x/m)\ge x$, which is the assertion. So the minimum is attained at
some $b>H$, where $U_{a^\ast}(b)=b$ and $U_{a^\ast}'(b)=1$. This is smooth
fit at $b$, that is $\gamma(a^\ast)=\vk(b)$, which is \eqref{eq:master} for
$(a^\ast,b)$. Conversely, any solution $(a,b)$ of \eqref{eq:master} with
$U_a\ge g$ on $[H,\infty)$ has $U_a(b)=b$ with $b>H$, so $h(a)=0$ and
$a=a^\ast$; this is the uniqueness.

\emph{Continuity.} On $(H,\infty)$, $U_a(x)-x=\gamma_1x^{\alpha_1}+\gamma_2x^{\alpha_2}-x$
with $\gamma_1,\gamma_2>0$ the components of $\gamma(a)$. As a function of
$x$ this is strictly convex, because $\alpha_1(\alpha_1-1)>0$ and
$\alpha_2(\alpha_2-1)>0$. Hence $b(H)$ is the unique minimiser of
$U_{a(H)}(x)-x$ on $[H,\infty)$. The pair $\gamma(a)$ depends continuously on
$(a,H)$, so $h$ is jointly continuous in $(a,H)$. For each fixed $H$ the map
$a\mapsto h(a,H)$ is strictly decreasing, so its zero $a(H)$ is unique; joint
continuity together with that strict monotonicity makes $H\mapsto a(H)$
continuous. The minimiser $b(H)$ is unique, by the strict convexity just noted, and it is
interior, $b(H)>H$, by the previous step, so that $U_{a(H)}'(b(H))=1$. The
function $x\mapsto U_{a(H)}(x)-x$ is jointly continuous in $(x,H)$ and
coercive, uniformly for $H$ in compact subsets of $(L,H_{\mathrm c}]$, so the
minimisers stay in a compact set and $H\mapsto b(H)$ is continuous as well. The argument applies on $(L,H_{\mathrm c}]$, the endpoint
included: at $H=H_{\mathrm c}$, Proposition~\ref{prop:merger} identifies
$U_{a_1}$ with the single-skew candidate at $H$ on $[a_2,b_2]$, where it
dominates $g$ with equality at $b_2$, while beyond $b_2$ its coefficient pair
is $\vk(b_2)$, so it dominates $g$ there too. Hence $h(a_1)=0$ at
$H=H_{\mathrm c}$, and $(a(H),b(H))\to(a_1,b_2)$ as $H\uparrow H_{\mathrm c}$.
\end{proof}

\section{Numerical validation}
\label{app:num}

Throughout $r=0.05$, $\mu=0.01$, $\sigma^2=0.04$, $L=1$. All closed forms were
checked against an independent assembly of the boundary conditions; the
same scripts produced Figs.~\ref{fig:regimes} and~\ref{fig:merger}, and are
available from the author.

\paragraph{Regime A and the existence of a band.}
Solving \eqref{eq:Asystem} for a single skew at $L$:

\begin{center}
\begin{tabular}{@{}lll@{}}
\toprule
$\beta_L$ & band $(a,b)$ & width\\
\midrule
$-0.30$ & no admissible band & ---\\
$0$     & none: sell immediately & $0$\\
$0.05$  & $(0.975163,\,1.025150)$ & $0.049987$\\
$0.30$  & $(0.856907,\,1.154180)$ & $0.297273$\\
$0.60$  & $(0.731329,\,1.311287)$ & $0.579958$\\
$0.90$  & $(0.625261,\,1.465402)$ & $0.840141$\\
\bottomrule
\end{tabular}
\end{center}

\noindent
confirming Corollary~\ref{cor:noband}. The system is invariant under
$(\beta,a,b)\mapsto(-\beta,b,a)$, which is why the only solution at
$\beta_L=-0.30$ is the swap of the $\beta_L=+0.30$ pair, with $a>b$, and is
therefore inadmissible.

\paragraph{Theorem~\ref{thm:Airrelevant}.}
With $\beta_L=0.30$ and $H=3.0$, the solution is
\[
(a,b)=(0.85690699,\,1.15417976)
\]
for every $\beta_H\in\{0,-0.3,-0.6,-0.9\}$, identical to eight decimal places
and equal to the single-skew value above.

\paragraph{Theorem~\ref{thm:corner} and Proposition~\ref{prop:transition}.}
With $\beta_L=0.90$, tracking the regime-B branch by continuation in
$\beta_H$:

\begin{center}
\begin{tabular}{@{}lllll@{}}
\toprule
$H$ & branch ends at & $b$ there & $a$ there & $a$ from \eqref{eq:Ceq}\\
\midrule
$1.30$ & $\beta_H=-0.1496$ & $1.300254$ & $0.640912$ & $0.640912$\\
$1.10$ & $\beta_H=-0.5146$ & $1.100038$ & $0.740293$ & $0.740293$\\
$1.05$ & $\beta_H=-0.6722$ & $1.050038$ & $0.803918$ & $0.803918$\\
\bottomrule
\end{tabular}
\end{center}

\noindent
The Jacobian determinant of \eqref{eq:master} along the branch equals
$-0.41$, $-0.11$ and $-0.048$ respectively at the endpoints, bounded away from
zero and negative as \eqref{eq:detJpi} predicts. The closed form $\beta_H^\star=(V'(H^-)-1)/(V'(H^-)+1)$ of
Theorem~\ref{thm:corner} gives $-0.1498$, $-0.51463$ and $-0.67222$,
matching the continuation endpoints to four significant figures.

\paragraph{Theorem~\ref{thm:verification}.}
Over $400$ random draws of $(\beta_L,\beta_H,H)$, with $\beta_L$ uniform on
$(0.05,0.95)$, $\beta_H$ uniform on $(-0.95,-0.05)$ and $\log(H/L)$ uniform on
$(0.02,0.9)$, yielding $301$ cases in regime~A, $18$ in~B and $81$ in~C, there
were no violations of $V\ge g$ on $(a,b)$ and no
local-time defect of the wrong sign. The monotonicity established in the proof
of Proposition~\ref{prop:dominance}, namely $W$ increasing on $(a,L)$ and
decreasing on $(L,b)$, held in every case, including in regime B across the
kink at $H$.

\paragraph{Proposition~\ref{prop:transition}.}
The closed form \eqref{eq:detJ} for $\det J$ agrees with finite differences to
a relative error below $10^{-9}$ at every point of the branch tested. The
monotonicity established in the proof is visible directly: with
$\beta_L=0.90$, $L=1$, $H=1.10$, the ratio $\pi$ rises strictly from
$0.32045844$ at $b=H$ to $1$ at $b=\bar b=1.465402$, so that $\beta_H$ sweeps
$(-0.51462548,\,0)$; and $(dB/da)/(dA/da)$ stays in $(0.38,0.45)$ while
$(dD_H/db)/(dG_H/db)$ is everywhere negative.

\paragraph{Section~\ref{sec:samesign}.}
With $\beta_L=\beta_H=0.60$ and $L=1$, the $L$-band is
$(0.731329,1.311287)$ for every $H$ and every $\beta_H$ tested, confirming the
decoupling of Proposition~\ref{prop:decouple}. The critical separation is
$H_{\mathrm c}=1.79301900$, where the bands are
$(0.731329,1.311287)$ and $(1.311287,2.351162)$ and the merged system returns
$(0.731329,2.351162)$, agreeing to six decimals on both endpoints. At $H=2.2$
the merged system still has the solution $(0.9418,2.8438)$, but at $x=1.460$
it takes the value $1.3247$ against a pay-off of $1.4603$, so it is not a
majorant. Sweeping $(\beta_L,\beta_H)\in\{0.3,0.6,0.8,0.9\}^2$ and
$H\in[1.10,1.75]$, the merged candidate dominates the pay-off at every
$H<H_{\mathrm c}$ and fails for $H>H_{\mathrm c}$; the first failure
in our sweep is $\beta_L=0.60$, $\beta_H=0.30$, $H=1.60$, where
$H_{\mathrm c}=1.5303$.

\paragraph{Theorem~\ref{thm:inversion}.}
The inversion \eqref{eq:inversion} round-trips exactly. Applied to the band
\[
(a,b)=(0.856907,\ 1.154180)
\]
generated by $\beta_L=0.30$ and $L=1$, it returns $L=1.000000$ and
$\beta_L=0.300000$. For $400$ random target bands
$(a,b)$ with $\log(b/a)$ uniform on $(0.05,3.5)$ --- ratios from $1.05$ to
$33.1$ --- and $r$ uniform on $(0.01,0.4)$, an admissible $(\sigma^2,\mu)$
giving $\beta_L\in(0,1)$ and $a<L<b$ was found in every case. The residuals
of \eqref{eq:Asystem} at the constructed parameters are at machine precision;
for the wide band $(0.4,12.0)$ with $r=0.09$, the choice
$\sigma^2=0.005$, $\mu/r=0.995$ gives $L=10.522634$,
$\beta_L=0.43295363$ with residuals $1.8\times10^{-15}$ and
$8.9\times10^{-16}$.

\end{document}